\documentclass[centertags,sumlimits,intlimits,namelimits]{amsart}
\usepackage{a4}
\usepackage{amsfonts,amsthm,amssymb}
\usepackage[pdftex,colorlinks,urlcolor=red,citecolor=blue,linkcolor=red]{hyperref}

\usepackage{chicago} 

\newcommand{\Bcal}{{\mathcal B}}

\newcommand{\Fcal}{{\mathcal F}}
\newcommand{\Gcal}{{\mathcal G}}
\newcommand{\Hcal}{{\mathcal H}}

\newcommand{\Ocal}{{\mathcal O}}
\newcommand{\Pcal}{{\mathcal P}}

\renewenvironment{description}{\list{}{
\leftmargin 12pt \itemindent8pt
}
}{
  \endlist
}

\newcommand{\half}{\frac{1}{2}}

\newtheorem{thm}{Theorem}[section]
\newtheorem{theorem}{Theorem}[section]
\newtheorem{prop}[thm]{Proposition}
\newtheorem{lem}[thm]{Lemma}

\newtheorem{corollary}[thm]{Corollary}

\newtheorem{ass}[thm]{Assumption}

\theoremstyle{definition}
\newtheorem{defin}[thm]{Definition}
\newtheorem{definition}[thm]{Definition}

\theoremstyle{remark}

\newtheorem{remark}[thm]{Remark}

\newcommand{\Ind}{{ 1}}
\newcommand{\ind}[1]{\Ind_{\{#1\}}}

\newcommand{\E}{\mathbb{E}}
\newcommand{\RR}{\mathbb{R}}
\newcommand{\PP}{\mathbb{P}}
\renewcommand{\P}{\PP}

\newcommand{\NN}{\mathbb{N}}
\newcommand{\N}{\mathbb{N}}

\newcommand{\cB}{{\mathcal B}}

\newcommand{\cF}{{\mathcal F}}

\newcommand{\cP}{{\mathcal P}}
\newcommand{\cE}{{\mathcal E}}

\newcommand{\cI}{{\mathcal I}}

\newcommand{\cT}{{\mathcal T}}

\newcommand{\stpr}[1]{(#1_t)_{t \ge 0}}

\renewcommand{\ind}[1]{\Ind_{\{#1\}}}

\newcommand{\squishlist}{
   \begin{list}{$\bullet$}
    { \setlength{\itemsep}{0pt}      \setlength{\parsep}{3pt}
      \setlength{\topsep}{3pt}       \setlength{\partopsep}{0pt}
      \setlength{\leftmargin}{1.5em} \setlength{\labelwidth}{1em}
      \setlength{\labelsep}{0.5em} }
      }

\newcommand{\Q}{\mathbb{Q}}

\newcommand{\R}{\RR}

\begin{document}

\hyphenation{Hil-bert}

\title[Existence and Monotonicity]{{\sf  Dynamic Term Structure Modelling with Default and Mortality Risk: New Results on  Existence  and Monotonicity}}
\author {Thorsten Schmidt}
\address{Technical University Chemnitz, Reichenhainer Str. 41, 09126 Chemnitz, Germany. Email:
thorsten.schmidt@mathematik.tu-chemnitz.de }
\author{Stefan Tappe}
\address{Leibniz Universit\"{a}t Hannover, Institut f\"{u}r Mathematische Stochastik, Welfengarten 1, 30167 Hannover, Germany. Email:
tappe@stochastik.uni-hannover.de}
\date{25 June, 2013}

\begin{abstract}
This paper considers general term structure models like the ones appearing in portfolio credit risk modelling or life insurance. We give a general model starting from families of forward rates driven by infinitely many Brownian motions and an integer-valued random measure, generalizing existing approaches in the literature. Then we derive drift conditions which are equivalent to no asymptotic free lunch on the considered market. Existence results are also given. In practice, models possessing a certain monotonicity are favorable and we study general conditions which guarantee this. The setup is illustrated with some examples.

\end{abstract}

\maketitle

\section{Introduction}

Numerous works in the literature study infinite-dimensional bond markets, with or without credit risk as for example  \citeN{BMKR}, \citeN{Filipovic2001}, \citeN{SchmidtOezkan}, \citeN{EkelandTaflin2005}, \citeN{TSchmidt_InfiniteFactors},   \citeN{CarmonaTehranchi2006},  \citeN{JakubowskiZabczyk:EMHJM}, \citeN{BarskiJakubowskiZabczyk2011}, and \citeN{BarsikZabczyk2012}  among many others.

In this paper we study a general account of such markets: we consider bond prices of the form
$$ P(t,T,\eta) $$
where $t \le T$ is current time, $T$ denotes the maturity and $\eta \in \cI$ denotes a quality index. This can refer to the credit quality of a term structure model, as it is the case in so-called rating approaches (see \citeN{JLT},  \citeN{BieleckiRutkowski00}, and \citeN{eberlein.oezkan03}). Or it could take the role of the number of already occurred losses in the context of credit portfolio modelling and collateralized debt obligations (see \citeN{FilipovicSchmidtOverbeck} and references therein). In the context of life insurance, $\eta$ denotes the age of the considered individual and models the effect that a higher age influences the survival probability as in \citeN{TappeWeber2013}. Also models for market or liquidity impacts have a similar structure, compare the recent approach in \citeN{JarrowRoch2013}.

Essentially we only assume that bond prices are non-negative and have a weak regularity in $T$ and $\eta$. This allows us to consider forward rate models in the Musiela parameterization, i.e.
$$ P(t,T,\eta)= I(t) \exp\bigg( \int_0^{T-t} r(t,u,\eta) du \bigg); $$
here $I$ is an indicator being zero when the bond prices are zero and one otherwise.
This is the starting point for modelling $r$ as Hilbert-space valued stochastic process given by a stochastic partial differential equation (SPDEs). The market of bond prices is certainly a large financial market which allows us to utilize the well-known concept of no asymptotic free lunch (NAFL) introduced in \citeN{Klein2000}. In this regard, the market satisfies an appropriate formulation of NAFL if and only if there exists an equivalent local martingale measure (ELMM). 

Conditions which render a measure an ELMM are following the approach of \citeN{HJM} and give the drift in terms of the volatility, which is why they are called drift conditions. In our case it turns out that two conditions are needed, one is a generalization of the classical drift condition to our more general setup and the second one links the instantaneous rate earned by the bond to the instantaneous risk and the risk-free rate. It turns out that this second condition makes it difficult to obtain explicit models. 

In this regard we consider a special setting where we are able to obtain existence results in our setup. We proceed in two steps: first we use the results obtained in \citeN{Jacod75} for martingale problems for marked point processes to obtain existence of a driving quality index process. Then we employ techniques from \citeN{Tappe2012a} to obtain conditions which guarantee existence of a unique mild solution of the SPDE for $r$ when the drift condition is satisfied. 

From a practical viewpoint it is natural, that a bond with a lower quality should be cheaper than a bond with higher quality. This is in general not implied by absence of arbitrage, as we discuss. However, if interest rates are non-negative, or more general, bond prices are martingales, the fact that the payoffs of the bonds are monotone in terms of the credit quality immediately gives monotonicity as the expectation is a monotone operator.  
More generally, we show that if the model is positivity preserving, then monotonicity follows. Finally we give sufficient conditions which yield positivity preserving term structures. Related results appear in \citeN{Barski2013}.

The organisation of the paper is as follows: after introducing the setup in Section 2, we discuss in detail the concept of absence of arbitrage considered in this paper. In particular, we derive the mentioned drift conditions.
Section 3 covers the existence results while Section 4 deals with positivity and monotonicity. In Section 5 we give a number of examples which illustrate the results.

\section{Arbitrage-free term structure movements}

\label{sec:generalsetup}
Consider a filtered probability space $(\Omega,\cF,(\cF_t)_{t \ge 0},\P)$ where the filtration $(\cF_t)_{t \ge 0}$ satisfies the usual conditions, i.e.\ it is right-continuous and $\cF_0$ contains all nullsets of $\cF$. 

We consider a market where bonds are traded. 
A $T$-\emph{bond} is a contingent claim which promises the payoff of one unit of currency at maturity $T$. We denote the price of the $T$-bond  at time $t \le T$ by $P(t,T)$. The bond is called \emph{risk-free} if $\P(P(T,T)=1)=1$. 
The stochastic process $(P(t,T)_{0 \le t \le T})$ describes  the evolution of the $T$-bond over time. 

In contrast to risk-free bonds we consider a more general framework where bonds carry an additional quality index $\eta$. There are bonds with different levels of quality. In this regard we consider a family of term structure models 
$$ \big\{(P(t,T,\eta))_{0 \le t \le T}: T \geq 0, \eta \in \cI\big\} $$
with some interval or countable set  $\cI\subset \R$. For our purposes, the typical choice will be  $\cI=[0,1]$.
The  index $\eta$ is called \emph{quality} of the bond. A bond with maturity $T$ and quality $\eta$  is called $(T,\eta)$-bond. 
Besides referencing to credit risk, the quality index can also refer to the liquidity of the bond or to the age of an individual in a life-insurance context, see Section \ref{examples} for examples.
This kind of term structures certainly play a central role in modelling multiple yield curves  or term structures for markets of collateralized debt obligations (CDOs) as we explain in Section \ref{sec:examples}. It turns out, that a basic tool for more involved models  is  to consider bonds with payoff zero or one which we treat here. 

\subsection{Absence of arbitrage in infinite dimensional bond markets}
The considered market of $(T,\eta)$-bonds is a market which contains an infinite number of traded assets. We view this setup in spirit of large financial markets and introduce the right concept of no-arbitrage in our setup which is  \emph{no asymptotic free lunch} (NAFL). This concept has been applied to bond markets in \citeN{KleinSchmidtTeichmann2013} and we generalize their work to our setting where bond prices are also indexed by credit quality. 

In this chapter we fix a finite time horizon $T^*>0$. Denote by $D=(D_t)_{0 \le t \le T^*}$ the risk-free bank account which is a non-negative, adapted process with $D(0)=1$. We will need the following assumption on continuity of the bond prices in $T$ and $\eta$ and on uniform local boundedness of bond prices and on local boundedness of the discounting factor.
\begin{ass}\label{ass1}
There is $N\in\cF$ with $\P(N)=0$ such that $N_1 \cup N_2\subset N$ where
\begin{align*}
N_1 &:= \bigcup_{t \in [0,T^*], \eta \in \cI} \big\{ \omega: T \to P(t,T,\eta)(\omega) \text{ is not absolutely continuous}\big\}, \\
N_2 &:= \bigcup_{t \in [0,T^*], t \le T \le T^* } \big\{ \omega: \eta \to P(t,T,\eta)(\omega) \text{ is not right continuous}\big\}.
\end{align*} 
\end{ass}
Note that in classical HJM-models absolute continuity in maturity always holds, such that $\P(N_1)=0$. We need furthermore right-continuity in the quality $\eta$ of the term structure models. 

\begin{ass}\label{ass2} The following holds:
\begin{enumerate}\renewcommand{\labelenumi}{(\roman{enumi})\ }
\item For any $(T,\eta)$ there is $\epsilon>0$, an increasing sequence of stopping times
$\tau_n\to\infty$ and $\kappa_n\in [0,\infty)$ such that
$$P(t,U,\xi)^{\tau_n}\leq \kappa_n,$$ for all $U\in[T,T+\epsilon)$, $\xi \in [\eta,\eta+\epsilon) \cap \cI$ and all
$t\leq T$,
\item $(D(t)) _{0\leq t\leq T^*}$ is locally bounded. 
\end{enumerate}
\end{ass}

\begin{definition}\label{LFM}
Fix a sequence $(T_i)_{i\in\N}$ in $[0,T^*]$ and a set $(\eta_i)_{i \in \N}\subset \cI$.
Define the $n^2+1$-dimensional stochastic process  $(\mathbf S^n)=(S^{0},S^{11},\dots,S^{nn})$ as follows: 
\begin{align}S^{ij}_t=
D(t) P(t \wedge T_i ,T_i, \eta_j),\qquad 0 \le t \le T^*, 
 \label{defSi}
\end{align}
for $i,j=1,\dots,n$ and $S^0_t\equiv 1$. The large financial market consists of the  sequence of classical markets $(\mathbf S^n)$.
\end{definition}
Note that this assumption is fullfilled when the family of term structure models is locally bounded and non-decreasing in $(T,\eta)$, a criterion for which we derive sufficient conditions in Section \ref{sec-pos-mon}.

In large financial markets absence of arbitrage is considered for each finite market $\mathbf{S}^n$ and appropriate limits. In this way we are able to avoid using measure-valued strategies, as for example used in \citeN{DeDonnoPratelli2005}.
Let $\mathbf{H}$ be a predictable $\mathbf{S}^n$-integrable process and denote by $(\mathbf{H}\cdot \mathbf{S}^n)_t$ the stochastic integral
of $\mathbf{H}$ with respect to $\mathbf{S}^n$ until $t$. The process $\mathbf{H}$ is called \emph{admissible
trading strategy} if $\mathbf{H}_0=0$ and there is $a>0$ such that $(\mathbf{H}\cdot \mathbf{S}^n)_t\ge -a$, $0 \le t \le T^*$.
Define the following cones:
\begin{equation}
\mathbf{K}^n =\{(\mathbf{H}\cdot \mathbf{S}^n)_{T^*}:\text{$H$ admissible}\}\text{ and }
\mathbf{C}^n =(\mathbf{K}^n-L^0_+)\cap L^{\infty}.\label{K}
\end{equation}
$\mathbf{K}^n$ containes all replicable claims in the finite market $n$, and $\mathbf{C}^n$ contains all claims in $L^{\infty}$
which  can be superreplicated. We define the set $\mathbf{M}_e^n$ of equivalent separating measures for the finite market $n$ as
\begin{align}\label{Me}
 \mathbf{M}_e^n &=
 \{\Q\sim \P|_{T^*}: \mathbf S^n \text{ is local $\Q$-martingale}\}
 \end{align}
 If $\mathbf{S}^n$ is bounded then $\mathbf{M}_e^n$ consists of all equivalent probability measures such that $\mathbf{S}^n$ is a (true) martingale.

We assume that for each finite market $n$ no arbitrage holds, i.e.
\begin{equation}\label{emm}
\mathbf{M}^n_e\ne\emptyset,\quad\quad\text{ for all $n\in\N$}.
\end{equation}
However, there is still the possibility of approximations
of an arbitrage profit by trading on the sequence of market models and we arrive at 
the following formulation.
\begin{definition}\label{N(A)FL} A given large financial market satisfies NAFL if
$$
\overline{\bigcup_{n=1}^{\infty}\mathbf{C}^n}^*\cap L^{\infty}_+ =\{0\}.
$$
\end{definition}

\begin{definition}\label{NAFL}
The family of term structure models
$ \{(P(t,T,\eta))_{0 \le t \le T}: T \geq 0, \eta \in \cI\}$ 
 satisfies NAFL if there exists a dense sequences
$(T_i)_{i\in\N}$ in $[0,\infty)$ and $(\eta_i)_{i \in \N}$ in $\cI$, such that  the large financial market of Definition~\ref{LFM} satisfies the condition NAFL.
\end{definition}

Inspection of  the proof  in  \citeN{KleinSchmidtTeichmann2013}, Theorem 5.2, shows that the following result holds in our case.
\begin{theorem}\label{th1}
Assume that Assumptions \ref{ass1}, \ref{ass2} and \eqref{emm} hold.  
The family of term structure models $ \{(P(t,T,\eta))_{0 \le t \le T}: 0 \le T \le T^*, \eta \in \cI\}$  satisfies NAFL, if and only if there exists a measure $\Q^*\sim \P|_{T^*}$ such that
\begin{align}\label{EMM}
(D_{t} P(t,T,\eta))_{0 \leq t \leq T} \ \text{are local $\Q^*$-martingales for all } (T,\eta) \in [0,T^*] \times \cI.
\end{align}
\end{theorem}

Such a measure $\Q^*$ is called equivalent local martingale measure (ELMM). In the  following section we derive drift conditions in spirit of the classical Heath-Jarrow-Morton drift condition which give \eqref{EMM} for arbitrary $T^*$.

\subsection{The considered term structure models}

As is customary in term structure modelling we directly consider the filtered probability space $(\Omega,\cF,(\cF_t)_{t \ge 0},\Q)$ with $\Q\sim\P$. The aim of the following sections is to give a precise setting of the considered term structure models under $\Q$ and thereafter derive conditions which are equivalent to NAFL.

In line with credit risk models, we associate a stopping time $\tau_\eta$ to each quality level $\eta$ and we assume that the payment of the $(T,\eta)$-bond takes place only if $\tau_\eta > T$.
It will be convenient to represent the model in terms of forward rates, such that we make the weak assumption that  $(T,\eta)$-bonds can be represented  by
\begin{align} \label{eq:TxbondsViaForwardRates}
 P(t,T,\eta) = \ind{\tau_\eta > t} \exp \bigg(-\int_t^T f(t,u,\eta) du \bigg);
 \end{align}
 $f(t,T,\eta)$ is called $(T,\eta)$-\emph{forward rate} at time $t$.

Our aim is to consider general, infinite-dimensional models for $f$.
In the spirit of \citeN{BMKR} and \citeN{CarmonaTehranchi}  we assume that $(T,\eta)$-forward rates  follow a semimartingale of the form
\begin{align} \label{def:f}
    df(t,T,\eta)  &= \alpha(t,T,\eta) dt +  \sigma(t,T,\eta) d W_t  \\
                  &+ \int_E \gamma(t,T,\eta,x) ( \mu(dt,dx) - F_t(dx)dt ), \qquad 0 \le t \le T; \nonumber
\end{align}
here $W$ is a $Q$-Wiener process on a separable Hilbert space $U$ with a trace class operator $Q \in L(U)$ (see \citeN{DaPratoZabczyk}) 
 and $\mu$ is a  integer-valued random measure on $\R_+ \times E$ with absolutely continuous compensator $dt \otimes F_t(dx)$ and
$E$ is the mark space. The mark space is a measurable space $(E,\mathcal{E})$  which we assume to be a Blackwell space (see \citeN{DellacherieMeyer}).
We remark that every Polish space with its Borel $\sigma$-field is a Blackwell space.

We assume that the stopping times $\tau_\eta$ have the following representation in terms of the random measure $\mu$. By $\Pcal$ we denote the predictable $\sigma$-algebra on $\Omega \times \R_+$. 
\begin{itemize}\setlength{\itemindent }{-2.5mm}
\item[{\bf(A1)}]  There is a $\mathbb{R}$-valued, $\mathcal{P} \otimes \mathcal{B}(\cI)  \otimes \cE$-measurable process $\beta$ such that 
\begin{align} \label{eq:tau}
  \ind{\tau_\eta >t } = 1+\int_0^t \int_E \ind{\tau_\eta \ge s} \beta(s,\eta,x) \mu(ds,dx).
\end{align}
\end{itemize}
\begin{remark}
Set $Y_t := \ind{\tau_\eta > t}$. Then  the representation \eqref{eq:tau} reads
$$  Y_t = 1+ \int_0^t \int_E  Y_{s-} \beta(s,\eta,x) \mu(ds,dx). $$
In Section \ref{sec:infinitebondmarket} we show how such a representation can be obtained in infinite
dimensional bond markets. 
\end{remark}

Recall that the separable Hilbert space $U$ denotes the state space of the Wiener process $W$.
Then there exists an orthonormal basis $(e_j)_{j \in \mathbb{N}}$ of $U$ and a sequence $(\lambda_j)_{j \in \mathbb{N}} \subset (0,\infty)$ with $\sum_{j \in \mathbb{N}} \lambda_j < \infty$ such that for all $u \in U$
  \begin{align*}
      Qu = \sum_{j \in \mathbb{N}} \lambda_j \langle u,e_j \rangle_{U} \, e_j;  
  \end{align*}
the $\lambda_j$ are the eigenvalues of $Q$, and each $e_j$ is an eigenvector corresponding to $\lambda_j$. The space $U_0 := Q^{1/2}(U)$, equipped with the inner product
\begin{align*}
\langle u,v \rangle_{U_0} := \langle Q^{-1/2} u, Q^{-1/2} v \rangle_{U}, 
\end{align*}
is another separable Hilbert space
and $( \sqrt{\lambda_j} e_j )_{j \in \mathbb{N}}$ is an orthonormal basis.
According to \citeN[Prop. 4.1]{DaPratoZabczyk}, the sequence of stochastic
processes $( W^j )_{j \in \mathbb{N}}$ defined as $W^j :=
\frac{1}{\sqrt{\lambda_j}} \langle W, e_j \rangle$ is a sequence of
real-valued independent Brownian motions and we
have the expansion
\begin{align*}
W = \sum_{j \in \mathbb{N}} \sqrt{\lambda _j} W^j e_j.
\end{align*}
Given another separable Hilbert space $H$, we denote by $L_2^0(H) := L_2(U_0,H)$ the space of Hilbert-Schmidt
operators from $U_0$ into $H$, which, endowed with the
Hilbert-Schmidt norm
\begin{align*}
\| \Phi \|_{L_2^0(H)} := \sqrt{\sum_{j \in \mathbb{N}} \lambda_j \| \Phi e_j \|^2},
\quad \Phi \in L_2^0(H)
\end{align*}
itself is a separable Hilbert space. Note that $L_2^0(H) \cong \ell^2(H)$, because $\Phi \mapsto (\Phi^j)_{j \in \mathbb{N}}$ with $\Phi^j := \sqrt{\lambda_j} \Phi e_j$, $j \in \mathbb{N}$ is an isometric isomorphism. According to \citeN[Thm. 4.3]{DaPratoZabczyk}, for every predictable process $\sigma : \Omega \times
\mathbb{R}_+ \rightarrow L_2^0(H)$ satisfying
\begin{align*}
\mathbb{P} \bigg( \int_0^t \| \sigma_s \|_{L_2^0(H)}^2 ds < \infty
\bigg) = 1 \quad \text{for all $t \geq 0$}
\end{align*}
with $\sigma^j_t:= \sqrt{\lambda_j} \sigma_t e_j$ we have the identity
\begin{align*}
\int_0^t \sigma_s dW_s = \sum_{j \in \mathbb{N}} \int_0^t \sigma_s^j d W_s^j, \quad t \geq 0.
\end{align*}
In particular, the diffusion term in (\ref{def:f}) can be written as
\begin{align*}
\sigma(t,T,\eta) d W_t = \sum_{j \in \mathbb{N}} \sigma^j(t,T,\eta) d W_t^j,
\end{align*}
where $\sigma^j(t,T,\eta)= \sqrt{\lambda_j} \sigma(t,T,\eta) e_j$.

Intensity-based default models correspond well with the assumption on absolute continuity of bond prices (which is implicit in \eqref{eq:TxbondsViaForwardRates}) which will enable us to obtain a drift condition in Theorem \ref{thm1}.
Set 
$$ \lambda(t,\eta) := - \int_E \beta(t,\eta,x) F_t (dx) . $$
The process $(1_{\{ \tau_\eta > \cdot \}})$ is decreasing and hence the Doob-Meyer decomposition gives a unique representation in terms of a local martingale and an absolutely continuous process. We denote this absolutely continuous process
by $\int_0^t \lambda(s,\eta) ds$, the \emph{cumulative intensity}. Then   
\begin{align}\label{eqMx}
    M^\eta_t &:= 1_{\{ \tau_\eta > t\}} + \int_0^t \ind{\tau_\eta \ge s} \lambda(s,\eta) ds
\end{align}
 is the (local) martingale in the Doob-Meyer decomposition. The non-negative process $(\lambda(t,\eta))_{t \ge 0}$ is called  (local) \emph{intensity} of $\tau_\eta$.

\subsection{Musiela parametrization}
It will be more convenient to consider the alternative parametrization
$$ r_t(\xi,\eta) := f(t,t+\xi,\eta) $$
which goes back to \citeN{Musiela}. Then $r$ is one single stochastic process with values in a function space $H$ of curves
$h:\R_+ \times \cI \to \R$ to be specified later. Assume that $r$ is continuous in $\xi$ and denote by $(S_t)_{t \ge 0}$
the shift semigroup on $H$, i.e.\ $S_th(\xi,\eta) = h(\xi+t,\eta)$. Then
equation \eqref{def:f} can be written as the variation of constants formula
\begin{align}\label{eq:rlong}
r_t(\xi,\eta) = S_t r_0(\xi,\eta)&+ \int_0^t S_{t-s} \alpha(s,s+\xi,\eta) ds +  \int_0^t S_{t-s} \sigma(s,s+\xi,\eta) d W_s \nonumber\\
&+\int_0^t \int_E S_{t-s}\gamma(s,s+\xi,\eta,x) (\mu(dt,dx)-F_s(dx)ds);
\end{align}
here $r_0\in H$ denotes the initial value of $r$ and $S_{t-s}$ operates on the functions
$\xi \mapsto \alpha(s,s+\xi,\eta)$, $\xi \mapsto \sigma(s,s+\xi,\eta)$, and $\xi \mapsto \gamma(s,s+\xi,\eta,x)$. In the following,
we will suppress the dependence on $(\xi,\eta)$. We obtain that \eqref{eq:rlong} can be written equivalently as the mild solution of
\begin{align}\label{dyn:r}
d r_t = \bigg( \frac{d}{d\xi} r_t + \alpha_t \bigg) dt  + \sigma_t d W_t + \int_E \gamma_t(x) \,  (\mu(ds,dx)-F_s(dx)ds).
\end{align}
In the following we denote by $\bar \mu(ds,dx):= \mu(ds,dx)-F_s(dx)ds$ the compensated random measure.

We make the following technical assumptions: Denote by
$\Ocal$ and $\Pcal$ the optional and predictable
$\sigma$-algebra on $\Omega\times\R_+$, respectively. Set $\cT:= \R_+\times \cI$.
\begin{itemize}
  \item[{\bf(A2)}] The initial curve $r_0$ is
  $\Bcal(\R_+)\otimes\Bcal(\cI)$-measurable, and locally integrable:
  \[ \int_0^\xi  |r_0(u,\eta)|\,du <\infty\quad\text{for all
  }(\xi,\eta) \in \cT, \quad \text{$\Q$--almost surely.}\]

  \item[{\bf(A3)}] The {\emph{drift}}  $\alpha_t(\xi,\eta)$ is $\R$-valued,
  $\Ocal\otimes\Bcal(\R_+)\otimes\Bcal(\cI)$-measurable, and locally
integrable: \[
  \int_0^\xi \int_0^\xi|\alpha_t(u,\eta)| \,du\, dt<\infty\quad\text{ for all $(\xi,\eta) \in \cT$,} \quad \text{$\Q$--almost surely.}\]

  \item[{\bf(A4)}] The \textit{volatility} $\sigma_t(\xi,\eta)$ is $L_2^0(\mathbb{R})$-valued, $\mathcal{O} \otimes \mathcal{B}(\mathbb{R}_+) \otimes \mathcal{B}(\cI)$-measurable, and locally square integrable:
\begin{align*}
\mathbb{E} \bigg[ \int_0^{\xi} \int_0^{\xi} \| \sigma_t(u,\eta) \|_{L_2^0(\mathbb{R})}^2 du dt \bigg] < \infty \quad\text{ for all $(\xi,\eta) \in \cT$.}
\end{align*}

  \item[{\bf(A5)}] The \textit{jump}-term $\gamma_t(x)(\xi,\eta)$ is $\mathbb{R}$-valued, $\mathcal{P} \otimes \mathcal{E} \otimes \mathcal{B}(\mathbb{R}_+) \otimes \mathcal{B}(\cI)$-measurable, and locally square integrable:
\begin{align*}
\mathbb{E} \bigg[ \int_0^{\xi} \int_0^{\xi} \int_E | \gamma_t(x)(u,\eta) |^2 F_t(dx) du dt \bigg] < \infty \quad\text{ for all $(\xi,\eta) \in \cT$.}
\end{align*}
\end{itemize}
Conditions {(A2)--(A5)} assert that the risk free {\em short
rate} $r_t=r_t(0,1)$ has a progressive version and satisfies $\int_0^T
|r_t|\,dt<\infty$ for all $T$, see e.g.\ \citeN{Filipovic2001}. The \emph{discounting process} in our setup is given by
$$ D_{t} =  e^{-\int_{0}^t r_sds}, \qquad t \ge 0. $$

\subsection{The drift conditions}
This section will derive drift conditions which ensure that the considered probability measure $\Q$ is an equivalent local martingale measure (ELMM). Then NAFL holds by Theorem \ref{th1}.

First, we introduce some notation. Let $A(t,T,\eta) := \int_0^{T-t} \alpha(t,s,\eta) ds$, $\Sigma^j(t,T,\eta) := \int_0^{T-t} \sigma^j(t,s,\eta) ds$
for all $j \in \mathbb{N}$,
and $\Gamma(t,T,\eta,x) := \int_0^{T-t} \gamma(t,s,\eta,x) ds$.

\begin{thm} \label{thm1}
Assume that (A1)--(A5) hold. Then $\Q$ is an ELMM, if and only if
\begin{align}\label{dc1}
      \alpha(t,T,\eta)  &= \sum_{j \in \mathbb{N}} \sigma^j (t,T,\eta) \Sigma^j(t,T,\eta) \\
&- \int _{E} \gamma(t,T,\eta,x) \Big( e^{-\Gamma(t,T,\eta,x) }(1+\beta(t,\eta,x))-1 \Big)  F_t(dx)\nonumber
 \\
    r_t(0,\eta) &= r_t+\lambda(t,\eta) , \label{dc2}
\end{align}
where \eqref{dc1} and \eqref{dc2} hold on $\{\tau_\eta > t\}$, $\Q\otimes dt$-a.s.
\end{thm}
In an auxiliary lemma we derive the dynamics of the pre-default bond prices and thereafter give the
proof of the theorem.
Let
$$ p(t,T,\eta) := \exp\bigg( - \int_0^{T-t} r_t(x,\eta) dx \bigg). $$
\begin{lem} \label{lem2.2}
Under (A2)--(A5) we have  for all $0 \le t \le T$ and $\eta \in \cI$ that
\begin{align*}
  d p(t,T,\eta) &= p(t-,T,\eta) m_t dt  +dM^{T,\eta}_t \end{align*}
  where $m_t$ equals
\begin{align*}
  r_t(0,\eta) -A(t,T,\eta) + \half  \sum_{j \in \mathbb{N}}\Sigma^j(t,T,\eta)^2
  + \int_E \Big(e^{-\Gamma(t,T,\eta,x)} -1 + \Gamma(t,T,\eta,x) \Big) F_t(dx)
  \end{align*} and $M^{T,\eta}$ is the local martingale given in \eqref{eq:Mdrift}.
\end{lem}
\begin{proof} The proof follows the arguments in \citeN{Filipovic2001}.
For $h\in H$ define $\cI_{T}:=\int_0^T h(s) ds$. We fix $\eta \in \cI$ and write
$$\cI_{T-t} r_t := \cI_{T-t} r_t(\cdot,\eta)=\int_0^{T-t} r_t(x,\eta) dx. $$
By the variation of constants formula \eqref{eq:rlong} we have that
\begin{align*}
  \cI_{T-t} r_t &= \cI_{T-t}(S_t r_0) + \int_0^t \cI_{T-t}(S_{t-u} \alpha_u) du + \int_0^t \cI_{T-t}(S_{t-u}\sigma_u) dW_u  \\
  &+ \int_0^t \cI_{T-t} (S_{t-u} \gamma_u(x)) \bar \mu(du,dx),
\end{align*}
where $\bar \mu$ is the compensated random measure. Note that $\cI_{T-t}(S_{t-u} h) = \cI_{T-u} h - \cI_{t-u} h$.
We apply this to all terms  and obtain $\cI_{T-t} r_t = I_1 - I_2$ with
\begin{align*}
I_1 &= \cI_{T} r_0 + \int_0^t \cI_{T-u} \alpha_u du + \int_0^t \cI_{T-u}\sigma_u dW_u  + \int_0^t \cI_{T-u}  \gamma_u(x) \bar \mu(du,dx),
\\ I_2 &= \cI_{t} r_0 + \int_0^t \cI_{t-u} \alpha_u du + \int_0^t \cI_{t-u}\sigma_u dW_u  + \int_0^t \cI_{t-u}  \gamma_u(x) \bar \mu(du,dx).
\end{align*}
We rearrange all the terms with the stochastic Fubini theorem according to the following argument:
\begin{align*}
  \int_0^t \cI_{t-u} \sigma_u dW_u &= \int_0^t \int_0^{t-u} \sigma_u(v) dv \, dW_u \\
  &= \int_0^t \int_u^t \sigma_u(v-u) dv \, dW_u \\
  &= \int_0^t \int_0^v \sigma_u(v-u) dW_u \, dv
\end{align*}
and obtain that
\begin{align*}
  I_2 &= \int_0^t \bigg( S_v r_0(0,\eta) + \int_0^v \Big(S_{v-u}  \alpha_u(0,\eta)du   \\ &\qquad\qquad\qquad\ \,+
  S_{v-u}  \sigma_u(0,\eta)dW_u+ \int_E S_{v-u}  \gamma_u(0,\eta,x) \bar \mu(du,dx) \Big)\bigg) dv \\
  &= \int_0^t r_v(0,\eta) dv.
\end{align*}
Hence, as $-\cI_T(r_0) = \ln p(0,T,\eta)$ we get $\Q$-a.s. and for all $0 \le t \le T$ that
\begin{align*}
\ln p(t,T,\eta) &= -\cI_{T-t}r_t = I_2 - I_1 \\
&=\ln p(0,T,\eta) + \int_0^t r_v(0,\eta) dv \\
&- \int_0^t \Big(\cI_{T-v}\alpha_v dv  + \cI_{T-v} \sigma_v dW_v + \int_E \cI_{T-v}\gamma_v(x) \bar \mu(dv,dx) \Big)
.
\end{align*}
Applying It\^o's formula to $e^x$ yields
\begin{align*}
\lefteqn{p(t,T,\eta) = p(0,T,\eta) + \half \sum_{j \in \mathbb{N}} \int_0^t p(v-,T,\eta) \big( \cI_{T-v} \sigma_v^j \big)^2 dv }\\
&+ \int_0^t p(v-,T,\eta) \bigg[ (r_v(0,\eta) - \cI_{T-v} \alpha_v )dv
- \cI_{T-v}\sigma_v dW_v - \int_E \cI_{T-v} \gamma_v(x) \bar \mu(dv,dx) \bigg] \\
&+ \int_E p(v-,T,\eta) \Big( e^{-\cI_{T-v}\gamma_v(x)}-1 + \cI_{T-v}\gamma_v(x) \Big) \mu(dv,dx) \\
&=  p(0,T,\eta) +  \int_0^t p(v-,T,\eta) \bigg[ \half\sum_{j \in \mathbb{N}}\big( \cI_{T-v} \sigma_v^j \big)^2
+r_v(0,\eta) - \cI_{T-v} \alpha_v  \\
& \qquad\qquad\qquad\qquad\quad\qquad\qquad+ \int_E\Big( e^{-\cI_{T-v}\gamma_v(x)}-1 + \cI_{T-v}\gamma_v(x) \Big) F_t(dx) \bigg] dv  \\
&+M^{T,\eta}_t,
\end{align*}
where $M^{T,\eta}$ are the local martingales
\begin{align}\label{eq:Mdrift}
\int_0^t  p(v-,T,\eta) \Big( - \cI_{T-v}\sigma_v dW_v - \int_E  ( e^{-\cI_{T-v}\gamma_v(x)}-1 ) \bar \mu(dv,dx) \Big).
\end{align}
Inserting the definitions of $A$, $\Sigma$ and $\Gamma$ we conclude.
\end{proof}

\begin{proof}[Proof of Theorem \ref{thm1}]
With the martingale $M^\eta$ from \eqref{eqMx},
\begin{align}
  d \big( D_t P(t,T,\eta)\big) &= d \big( (D_t p(t,T,\eta)) \cdot \ind{ \tau_{\eta} \geq t } \big) \\
      &= D_t p(t-,T,\eta) dM^\eta_t - D_t p(t-,T,\eta) \lambda(t,\eta) \ind{\tau_\eta\ge t} \, dt  \nonumber\\
      &+ \ind{\tau_\eta \ge t} d(D_t p(t,T,\eta))+ d[D p(\cdot,T,\eta), \ind{ \tau_{\eta} \geq \cdot }]_t. \label{temp644}
\end{align}
In the following, we compute all terms separately. 
First, as $D$ is of finite variation, the product rule and Lemma~\ref{lem2.2} give 
\begin{align}
  d (D_tp(t,T,\eta))
&= D_tp(t-,T,\eta) \bigg[  r_t(0,\eta) - r_t -  A(t,T,\eta) + \half \sum_{j \in \mathbb{N}}\Sigma^j(t,T,\eta)^2 \nonumber\\
& +\int _{E}\Big( e^{-  \Gamma(t,T,\eta,x) } -1+\Gamma(t,T,\eta,x) \Big) \,  F_t(dx) \bigg] dt \nonumber\\
& + d\tilde M_t \label{eq:dynamicsDP}
  \end{align}
  where $\tilde M$ is a local martingale.
Second, we recall from Lemma \ref{lem2.2}  that
$$ \Delta (D_t p(t,T,\eta)) = D_t \Delta p(t,T,\eta) = D_t p(t-,T,\eta) \int_E   \left(e^{-\Gamma(t,T,\eta,x)}-1 \right)\mu(dt,dx). $$
Assumption (A1) immediately gives 
$$ \Delta (\ind{\tau_\eta > t}) =  \int_E \ind{\tau_\eta \ge t} \beta(t,\eta,x) \mu(dt,dx) . $$
Altogether we obtain the quadratic covariation of discounted $(T,\eta)$-bond prices and the default indicator process,
\begin{align} \label{eq:jointjumps}
\lefteqn{d [D p(\cdot,T,\eta), \ind{\tau_\eta > \cdot}]_t}  \nonumber \\
&=  D_t  p(t-,T,\eta)   \ind{\tau_\eta \ge t} \int_E   \left(e^{-\Gamma(t,T,\eta,x)}-1 \right) \beta(t,\eta,x)    \mu(dt,dx)
\end{align}
and all terms in \eqref{temp644} have been computed. Note that $ p(t-,T,\eta)   \ind{\tau_\eta \ge t}=p(t-,T,\eta)$.
Finally, $\Q\in{\mathcal Q}$ if and only if $DP$ is a local martingale. 
The drift condition is now obtained by the fact that $DP$ is a local martingale if and only if its drift vanishes. 
On $\{\tau_\eta \le t\}$, $DP$ is zero and no drift condition applies. Otherwise, on $\{\tau_\eta >t\}$
we have $DP = Dp$. 
From \eqref{temp644}, \eqref{eq:dynamicsDP} and \eqref{eq:jointjumps}  we therefore obtain the following drift condition (as  $D_t p(t-,T,\eta)>0$):
\begin{align*}
0
&=     r_t(0,\eta) - r_t -\lambda(t,\eta) -  A(t,T,\eta)  +\half \sum_{j \in \mathbb{N}}\Sigma^j(t,T,\eta)^2 \\
&+ \int _{E} \Big( e^{-  \Gamma(t,T,\eta,x) } -1+\Gamma(t,T,\eta,x) \Big)  F_t(dx) \\
&+ \int_E   \left(e^{-\Gamma(t,T,\eta,x)}-1 \right)   \beta(t,\eta,x)  F_t(dx) \\
&=   r_t(0,\eta) - r_t -\lambda(t,\eta) -  A(t,T,\eta)  +\half \sum_{j \in \mathbb{N}}\Sigma^j(t,T,\eta)^2 \\
&+ \int_E \bigg[  \left(e^{-\Gamma(t,T,\eta,x)}-1 \right) (1+\beta(t,\eta,x))  +\Gamma(t,T,\eta,x)  \bigg] F_t(dx).
\end{align*}
First, letting $T=t$  we obtain \eqref{dc2}. Differentiating the remaining terms with respect to $T$ gives \eqref{dc1}.

For the converse, we need to show that  the drift conditions imply that all discounted digital options are local martingales.
For fixed $\eta$, these conditions imply that on $\{L_t\le \eta\}$ discounted prices are local martingales. On the other side,
on $\{L_t > \eta\}$ the prices are zero by definition and hence martingales. The conclusion follows.
\end{proof}

%
%

\section{Existence}
Existence in the general model of Section \ref{sec:generalsetup} is difficult to obtain. 
It turns out that in many applications, on can concentrate on the following special case, see Section \ref{examples} for appropriate examples. 

Consider a pure-jump process $L$ with values in $\cI$ which can be  interpreted as \emph{quality index} of the considered market. For simplicity we consider $\cI=[0,1]$. We allow for arbitrary granularity, i.e.~all $\tau_\eta := \inf\{t \ge 0: L_t > \eta\},\ \eta \in \cI$ are considered.

Assume that $E=G \times \cI$,  and $G$ is the mark space of a homogeneous Poisson random measure $\tilde \mu$, see \citeN[Def. II.1.20]{JacodShiryaev}. For  existence, homogeneity of $\tilde \mu$ is not relevant, but is simplifies the study of positivity and monotonicity in the following sections.

Denote by $\mu^L$ the Poisson random measure associated to the jumps of $L$, that is
\begin{align}\label{def:L2}
 L_t = \int_0^t \int_{\cI} x \, \mu^L(ds,dx). \end{align}
Letting $\tau_\eta := \inf\{t \ge 0: L_t >\eta \}$ we establish the link to $(T,\eta)$-bonds such that
\begin{align}\label{repr-bond}
P(t,T,\eta) = \ind{L_t \le \eta} \exp \bigg( - \int_t^{T} f(t,u,\eta) du\bigg).
\end{align}
 We set $\mu=\mu^L \otimes \tilde \mu$ and assume that
 $r$ is the mild solution of
\begin{align}\label{def:r2}
d r_t &= \bigg( \frac{d}{d\xi} r_t + \alpha_t \bigg) dt  + \sigma_t d W_t + \int_G \gamma_t(x) \,  \tilde \mu(dt,dx)
 + \int_{\cI} \delta_t(x) \mu^L(dt,dx).
\end{align}
Note that in contrast to \eqref{dyn:r}, $r$ is not given in terms of compensated measures. It will turn out that this leads
to a simplification in the drift conditions, which will be shown in the following proposition.
This setting generalizes the approach \citeN{FilipovicSchmidtOverbeck} in the way that it incorporates jumps in $r$ besides jumps induced by the loss process.
\citeN{SchmidtZabczyk12} study the case where $\tilde \mu$ is a L\'evy-process.

\noindent We adapt Assumptions (A1) and (A5) to this setting. 
\begin{itemize}
\item[{\bf(A1')}] $L_t= \sum_{s\le t} \Delta L_s$ is a c\`adl\`ag, non-decreasing, adapted, pure jump process with values in $\cI$
  which admits an absolutely continuous compensator $\nu^L(t,dx)dt$ satisfying $ \nu^L(t,\cI)<\infty$ for all $t \ge 0$.
  $\tilde \mu$ is a homogeneous Poisson random measure on $\R_+\times G$ with
compensator $dt \otimes \tilde F(dx)$ and $\tilde{F}(G) < \infty$. Moreover, $\tilde \mu$ and $\mu^L$ are independent.
\item[{\bf(A5')}] $\gamma_t(x)(\xi,\eta)$ is $\R$-valued $\Pcal\otimes\Bcal(G)\otimes\Bcal(\R_+)\otimes\Bcal(I)$-measurable,
  and locally square integrable:
  \begin{align*}
    \mathbb{E} \bigg[ \int_0^{\xi} \int_0^{\xi} \int_G | \gamma_t(x)(u,\eta) | \tilde{F}(dx) du dt \bigg] < \infty \quad \text{for all $(\xi,\eta) \in \mathcal{T}$.}
  \end{align*}
  $\delta_t(x)(\xi,\eta)$ is $\R$-valued $\Pcal\otimes\Bcal(I)\otimes\Bcal(\R_+)\otimes\Bcal(I)$-measurable,
  and locally square integrable:
  \begin{align*}
    \mathbb{E} \bigg[ \int_0^{\xi} \int_0^{\xi} \int_{\cI} | \delta_t(x)(u,\eta) | \nu^L(\eta,dx) du dt \bigg] < \infty \quad \text{for all $(\xi,\eta) \in \mathcal{T}$.}
  \end{align*}
\end{itemize}

We obtain absence of arbitrage, or more precisely NAFL, in this setting by a direct application of Theorem~\ref{thm1}. With the notation from this theorem we have that $\sigma^j(t,T,\eta)=\sqrt{\lambda_j}\sigma(t,T,\eta)e_j$ and $\Sigma^j(t,T,\eta)=\int_t^T \sigma^j(t,s,\eta) ds$. Set $\Delta(t,T,\eta,x) := \int_0^{T-t} \delta(t,s,\eta,x) ds.$
\begin{prop}\label{prop:2.1}
Under (A1'), (A2)--(A4) and (A5') and with $r$ as given in \eqref{dyn:r}, we have that $\Q$ is ELMM, if and only if
\begin{align}
      \alpha(t,T,\eta)  &= \sum_{j \in \mathbb{N}} \sigma^j (t,T,\eta) \Sigma^j(t,T,\eta) \label{dc1:2}\\
                        &- \int _{E} \gamma(t,T,\eta,x)  e^{-\Gamma(t,T,\eta,x) } \tilde  F(dx) \nonumber \\
                        &- \int _{\cI} \delta(t,T,\eta,x)  \ind{L_{t-}+x\le \eta} e^{-\Delta(t,T,\eta,x) }  \nu_t^L(dx),\nonumber
 \\
    r_t(0,\eta) &= r_t+\lambda(t,\eta) , \label{dc2:2}
\end{align}
where \eqref{dc1:2} and \eqref{dc2:2} hold on $\{L_t \le \eta\}$, $\Q\otimes dt$-a.s.
\end{prop}
\begin{proof}
Our aim is to apply Theorem \ref{thm1}.
We write  $x=(x_1,x_2)\in \cI\times G$ with $x_1 \in \cI$ and $x_2 \in G$. Then \eqref{def:L2} gives that
\begin{align*} 
 L_t = \int_0^t \int_E \ell_s(x) \mu(ds,dx), 
 \end{align*}
with $\ell_t(x):= x_1$ as $\mu=\mu^L \otimes \tilde \mu$. Note that, by definition $L$ takes its values in $\cI$.
Next, we need to obtain a representation of $\tau_\eta$ in terms of $L$.  
Note that   $\ind{\tau_\eta >t }= \ind{L_t \le \eta}$. Hence, by uniqueness of the Doob-Meyer decomposition, we obtain that
the compensators of this two processes must coincide, i.e.
\begin{align*} - \int_E \beta(t,\eta,x) F_t (dx) & = F_t ( \{x \in E: L_{t-} + \ell_t(x) > \eta \} ) \\
    &= \int_E \ind{L_{t-} + \ell_t(x) > \eta} F_t (dx).  
\end{align*}
This can be satisfied by choosing $\beta(t,\eta,x) := -\ind{L_{t-} + \ell_t(x) > \eta}$ such that
\begin{align}\label{tempbeta} 
1+ \beta(t,\eta,x) := -\ind{L_{t-} + \ell_t(x) \le \eta} . 
\end{align}

The next step is to derive the dynamics of $r$ given in \eqref{def:r2} in terms of \eqref{dyn:r}. In this regard we have that
\begin{align*}
d r_t &= \bigg( \frac{d}{d\xi} r_t + \alpha_t + \int_G \gamma_t(x) \tilde F(dx) + \int_I \delta_t(x) \nu^L_t(dx)\bigg) dt  + \sigma_t d W_t \\
& + \int_G \gamma_t(x) \,  (\tilde { \mu}(dt,dx) - \tilde F(dx) dt)  \\
& + \int_I \delta_t(x) ( \bar \mu^L(dt,dx)- \nu^L_t(dx)dt).
\end{align*}
A careful application of Theorem \ref{thm1} together with \eqref{tempbeta}  gives
\begin{align*}
     \lefteqn{ \alpha(t,T,\eta) + \int_G \gamma_t(T,\eta,x) \tilde F(dx) + \int_{\cI} \delta_t(T,\eta,x) \nu^L_t(dx)  }\qquad \\
       &= \sum_{j \in \mathbb{N}} \sigma^j (t,T,\eta) \Sigma^j(t,T,\eta) \nonumber\\
                        &- \int _{E} \gamma_t(T,\eta,x) \Big( e^{-\Gamma(t,T,\eta,x) }-1 \Big)  \tilde F(dx) \nonumber \\
                        &- \int _{\cI} \delta_t(T,\eta,x) \Big( e^{-\Delta(t,T,\eta,x) }\ind{L_{t-}+x\le \eta}-1 \Big)  \nu_t^L(dx)
\end{align*}
which yields \eqref{dc1:2} and we conclude.
\end{proof}

\subsection{A martingale problem}
The existence in our setting is a direct extension from  \citeN{FilipovicSchmidtOverbeck}, Theorem 5.1. However, the construction is important for the following results
and in this section we state the result in our setting.
We  assume that the stochastic basis satisfies:
\begin{itemize}
\item[{\bf(A6)}] $\Omega=\Omega_1\times\Omega_2$, $\Fcal=\Gcal\otimes\Hcal$,
$\Q(d\omega)=\Q_1(d\omega_1)\Q_2(\omega_1,d\omega_2)$, where
  $\omega=(\omega_1,\omega_2)\in\Omega$, and
$\Fcal_t=\Gcal_t\otimes \Hcal_t$,
  where
  \begin{enumerate}
  \item $(\Omega_1,\Gcal,(\Gcal_t),\Q_1)$ is some filtered probability
  space carrying the market information, in particular the Brownian
  motions $W^j(\omega)=W^j(\omega_1)$, $j \in \mathbb{N}$ and the Poisson random measure $\tilde \mu(\omega)=\tilde \mu(\omega_1)$,

  \item $(\Omega_2,\Hcal)$ is the canonical space of paths for $I$-valued increasing
  marked point processes endowed with the minimal filtration $(\Hcal_t)$:
  the generic $\omega_2\in\Omega_2$ is a c\`adl\`ag, increasing, piecewise
  constant function from $\R_+$ to $I$. Let
  \[L_t(\omega)=\omega_2(t) \] be the coordinate process.
  The filtration $(\Hcal_t)$ is therefore $\Hcal_t=\sigma(L_s\mid s\le t)$, and
  $\Hcal=\Hcal_\infty$,

  \item $\Q_2$ is a probability kernel from $(\Omega_1,\Gcal)$ to
  $\Hcal$ to be determined below.
  \end{enumerate}
\end{itemize}
Under assumption (A6), the volatility $\sigma_t(\omega)=\sigma_t(\omega_1,\omega_2)$, and the jump terms $\gamma_t(\omega;x)=\gamma_t(\omega_1,\omega_2;x)$ and $\delta_t(\omega;x )=\delta_t(\omega_1,\omega_2;x)$ in
{(A3)} and (A5') actually are functions of the loss path
$\omega_2$.  The evolution equation \eqref{def:r2} can thus be solved on
the stochastic basis $(\Omega_1,\Gcal,(\Gcal_t),\Q_1)$ along any
genuine loss path $\omega_2\in\Omega_2$. Indeed, the integral with
respect to $\mu$ in \eqref{def:r2} is path-wise in $\omega_2$.

Regarding condition \eqref{dc2:2}, note that under (A1'), the intensity
 $\lambda(t,\eta)$ uniquely determines $\nu^L(t,dx)$ via
  \begin{equation}\label{eqnudeflam}
    \nu^L(t,(0,\eta])=\lambda(t,L_t)-\lambda(t,L_t+\eta),\quad
    \eta\in I,
  \end{equation}
  where we denote $\lambda(t,x)=0$ for $x\ge 1$.
Then, condition \eqref{dc2:2} is
equivalent to
\begin{equation}\label{nudefdef}
  \nu^L(\omega;t,dx)=-r_t(\omega;0,\omega_2(t)+dx), \quad\text{(set
  $r_t(0,\eta)\equiv r_t$ for $\eta\ge 1$).}
\end{equation}
Hence, unless  $\delta$ is zero,
\[ \alpha_t(\xi,\eta)=\alpha_t(\xi,\eta,r_t) \]
in \eqref{dc1:2} becomes via \eqref{nudefdef} an explicit linear
functional of the (short end of the) prevailing spread curve. In
fact, there may result an implicit non-linear smooth dependence on
the entire prevailing spread curve $r_t$ via $\sigma$ and $\gamma$ in
\eqref{dc1:2}, respectively. But since this dependence on
$r_t$ is smooth, for any given loss path
$\omega_2\in\Omega_2$, equation \eqref{def:r2} will generically be
uniquely solvable.

It thus remains to find a probability kernel $\Q_2$ such that $\nu$
in \eqref{nudefdef} becomes the compensator of $L$. This is a
martingale problem for marked point processes, which has completely
been solved by  \citeN{Jacod75}. It turns out that $\Q_2$ exists
and is unique.

\begin{theorem}\label{thm2gen}
Assume {(A6)} holds. Let $r_0$, $\sigma_t, \gamma_t(x)$ and $\delta_t(x)$
satisfy {(A2)}, {(A4)} and {(A5')}, respectively. Define
$\nu^L(t,dx)$ by \eqref{nudefdef} and $\alpha_t$ by \eqref{dc1:2} for
all $(t,T,x)$.
Suppose, for any loss path $\omega_2\in\Omega_2$, there exists a
solution $r_t(\xi,\eta)$ of \eqref{def:r2} such that $r_t(0,\eta)$ is
progressive, decreasing and c\`adl\`ag in $\eta\in I$. Then
\begin{enumerate}
  \item\label{thm2gen1} {{(A3)}} is satisfied.

  \item\label{thm2gen2} There exists a unique probability
kernel $\Q_2$ from $(\Omega_1,\Gcal)$ to $\Hcal$, such that the loss
process $L_t(\omega)=\omega_2(t)$ satisfies {(A1')} and the
no-arbitrage condition \eqref{EMM} holds.
\end{enumerate}
\end{theorem}
%
%
The proof is analogous to Theorem 5.1 in \citeN{FilipovicSchmidtOverbeck}.

\subsection{An SPDE approach}
The next step will be to state  conditions which guarantee the existence of solutions of a SPDE as in \eqref{def:r2}. More precisely, 
we consider the SPDE
\begin{align}\label{Musiela-CDO}
\left\{
\begin{array}{rcl}
dr_t & = & \big( \frac{d}{d\xi} r_t + \alpha(r_t) \big) dt + \sigma(r_t)dW_t 
\\ && + \int_G \gamma(r_{t-},x) \tilde{\mu}(dt,dx) + \int_I \delta(r_{t-},x) \mu^L(dt,dx) \medskip
\\ r_0 & = & h_0
\end{array}
\right.
\end{align}
with measurable mappings
$\sigma : H_{\beta} \rightarrow L_2^0(H_{\beta})$, $\gamma : H_{\beta} \times G \rightarrow H_{\beta}$ and $\delta : H_{\beta} \times I \rightarrow H_{\beta}$ on a suitable Hilbert space $H_{\beta}$ consisting of functions $h : \mathbb{R}_+ \times [0,1] \rightarrow \mathbb{R}$, which we shall now introduce. Let $\beta > 0$ be an arbitrary constant.


\begin{defin}
Let $H_{\beta}$ be the linear space consisting of all functions $h : \mathbb{R}_+ \times [0,1] \rightarrow \mathbb{R}$ satisfying the following conditions:
\begin{itemize}
\item For each $\xi \in \mathbb{R}_+$ the mapping $h(\xi,\cdot)$, $\partial_{\xi} h(\xi,\cdot)$ are absolutely continuous, and for each $\eta \in [0,1]$ the mappings $h(\cdot,\eta)$, $\partial_{\eta} h(\cdot,\eta)$ are absolutely continuous (and hence, almost everywhere differentiable).

\item We have almost everywhere $\partial_{\xi \eta} h = \partial_{\eta \xi} h$.

\item We have
\begin{equation}\label{def-norm-beta}
\begin{aligned}
\| h \|_{\beta} &:= \bigg( |h(0,0)|^2 + \int_0^{\infty} | \partial_{\xi} h(\xi,0) |^2 e^{\beta \xi} d\xi
+ \int_0^1 | \partial_{\eta} h(0,\eta) |^2 d\eta
\\ &\quad\quad\, + \int_0^{\infty} \int_0^1 | \partial_{\xi \eta} h(\xi,\eta) |^2 d\eta e^{\beta \xi} d\xi \bigg)^{1/2} < \infty.
\end{aligned}
\end{equation}
\end{itemize}
\end{defin}

In contrast to default-free term structure modelling, the SPDE (\ref{Musiela-CDO}) describes the dynamics of two-dimensional surfaces rather than curves, which is due to the additional parameter $\eta$ that describes the quality of bonds. In the context of default-free term structure modelling, similar spaces consisting of curves have been introduced in \citeN{Filipovic2001}. Let us collect some relevant properties of the spaces $H_{\beta}$. The proof of the following result can be provided by using similar techniques as in \citeN[Section~5]{Filipovic2001}, \citeN[Section~4]{Tappe2010}, \citeN[Appendix]{FilipovicTappeTeichmann2010} and \cite[Appendix~A]{Tappe2012b}, and it is therefore omitted.

\begin{theorem}\label{thm-group}
Let $\beta > 0$ be arbitrary.
\begin{enumerate}
\item The linear space $(H_{\beta},\| \cdot \|_{\beta})$ is a separable
Hilbert space.

\item The shift-semigroup $(S_t)_{t \geq 0}$ given by
\begin{align*}
S_t : H_{\beta} \rightarrow H_{\beta}, \quad S_t h(\xi,\eta) = h(\xi + t,\eta)
\end{align*}
is a $C_0$-semigroup on $H_{\beta}$ with infinitesimal generator $d / d\xi$.

\item There exist another separable Hilbert space $\mathcal{H}_{\beta}$, a $C_0$-group $(U_t)_{t \in \mathbb{R}}$ on $\mathcal{H}_{\beta}$ and continuous linear operators $\ell \in L(H_{\beta},\mathcal{H}_{\beta})$, $\pi \in L(\mathcal{H}_{\beta},H_{\beta})$ such that $\pi U_t \ell = S_t$ for all $t \in \mathbb{R}_+$.

\item The linear space
\begin{align*}
H_{\beta}^0 = \Big\{ h \in H_{\beta} : \lim_{\xi \rightarrow \infty} h(\xi,0) = 0 \text{ and } \lim_{\xi \rightarrow \infty} \partial_{\eta} h(\xi,\eta) = 0 \text{ for all } \eta \in [0,1] \Big\}
\end{align*}
is a closed subspace of $H_{\beta}$.

\item Each function $h \in H_{\beta}$ is continuous and bounded.

\item For all $(\xi,\eta) \in \R_+ \times [0,1]$ the point evaluation $h \mapsto h(\xi,\eta) : H_{\beta} \rightarrow \R$ is a continuous linear functional.

\item There is a constant $C_1 > 0$, only depending on $\beta$, such that for all $h \in H_{\beta}$ we have
\begin{align}\label{est-C4} 
\| h \|_{\infty} &\leq C_1 \| h \|_{\beta}.
\end{align}

\item For all $h \in H_{\beta}$ we have $\exp(h) \in H_{\beta}$, there are constants $C_2,C_3 > 0$, only depending on $\beta$, such that for all $h \in H_{\beta}$ we have
\begin{align}\label{exp-lin-growth}
\| \exp(h) \|_{\beta} \leq C_2 (1 + \| h \|_{\beta}) \exp(C_3 \| h \|_{\beta}),
\end{align}
and the mapping $h \mapsto \exp(h)$ is locally Lipschitz continuous on $H_{\beta}$.

\item For all $h,g \in H_{\beta}$ we have $hg \in H_{\beta}$ and the multiplication map $m : H_{\beta} \times H_{\beta} \rightarrow H_{\beta}$ defined as $m(h,g) := hg$ is a continuous, bilinear operator.

\item Let $\beta' > \beta$ be arbitrary. Then we have $H_{\beta'} \subset
H_{\beta}$ and the estimate
\begin{align*}
\| h \|_{\beta} \leq \| h \|_{\beta'}, \quad h \in H_{\beta'}.
\end{align*}
Moreover, for each $h \in H_{\beta'}^0$ we have $\mathcal{I} h \in H_{\beta}$, where
\begin{align*}
\mathcal{I} h (\xi,\eta) := \int_0^{\xi} h(\zeta,\eta) d\zeta, \quad (\xi,\eta) \in \mathbb{R}_+ \times [0,1],
\end{align*}
and the integral operator $\mathcal{I} : H_{\beta'}^0 \rightarrow H_{\beta}$ is a continuous linear operator.

\end{enumerate}
\end{theorem}

In particular, we see that $H_{\beta}$ is a separable Hilbert space and that the shift semigroup $(S_t)_{t \geq 0}$ is a $C_0$-semigroup on $H_{\beta}$ with infinitesimal generator $d / d\xi$. In order to provide our existence- and uniqueness result, we impose the following conditions.

\begin{description}
\item[(A7)]
Let $\beta' > \beta$ be another constant. We assume that:
\begin{itemize}
\item $\sigma(H_{\beta}) \subset L_2^0(H_{\beta'}^0)$ and $\gamma(H_{\beta} \times G), \delta(H_{\beta} \times I) \subset H_{\beta'}^0$.

\item There is a sequence $(c^j)_{j \in \mathbb{N}}$ with $\sum_{j \in \mathbb{N}} (c^j)^2 < \infty$ such that for all $j \in \mathbb{N}$ we have
\begin{align}\label{sigma-cond-1}
\| \sigma^j(h) \|_{\beta'} &\leq c^j, \quad h \in H_{\beta},
\\ \label{sigma-cond-2} \| \sigma^j(h_1) - \sigma^j(h_2) \|_{\beta'} &\leq c^j \| h_1 - h_2 \|_{\beta}, \quad h_1,h_2 \in H_{\beta}.
\end{align}
\item There is a constant $M > 0$ such that for all $x \in G$ we have
\begin{align}\label{gamma-cond-1}
\| \gamma(h,x) \|_{\beta'} &\leq M, \quad h \in H_{\beta},
\\ \label{gamma-cond-2} \| \gamma(h_1,x) - \gamma(h_2,x) \|_{\beta'} &\leq M \| h_1 - h_2 \|_{\beta}, \quad h_1,h_2 \in H_{\beta},
\end{align}
and for all $x \in I$ we have
\begin{align}\label{delta-cond-1}
\| \delta(h,x) \|_{\beta'} &\leq M, \quad h \in H_{\beta},
\\ \label{delta-cond-2} \| \delta(h_1,x) - \delta(h_2,x) \|_{\beta'} &\leq M \| h_1 - h_2 \|_{\beta}, \quad h_1,h_2 \in H_{\beta}.
\end{align}
\end{itemize}
\end{description}

In view of Proposition~\ref{prop:2.1}, we suppose that the drift term
$\alpha : \Omega_2 \times \mathbb{R}_+ \times H_{\beta} \rightarrow H_{\beta}$ in the SPDE (\ref{Musiela-CDO}) is given by
\begin{equation}\label{SPDE:dc}
\begin{aligned}
\alpha(\omega_2,t,h)(\xi,\eta) &= \sum_{j \in \mathbb{N}} \sigma^j(h)(\xi,\eta) \Sigma^j(h)(\xi,\eta) 
\\ &\quad - \int_G \gamma(h,x)(\xi,\eta) e^{-\Gamma(h,x)(\xi,\eta)} \tilde{F}(dx)
\\ &\quad - \int_I \Ind_{ \{ \omega_2(t-) + x \leq \eta \} } \delta(h,x)(\xi,\eta) e^{- \Delta(h,x)(\xi,\eta)} h(0,\omega_2(t)+dx). 
\end{aligned}
\end{equation}
We will prove the following existence- and uniqueness result:

\begin{theorem}\label{4.3}
For each $h_0 \in H_{\beta}$ and each $\omega_2 \in \Omega_2$ there exists a unique mild solution $r = r(\cdot,\omega_2) : \Omega_1 \times \mathbb{R}_+ \rightarrow H_{\beta}$ to the SPDE
\begin{align}\label{SPDE-omega-2}
\left\{
\begin{array}{rcl}
dr_t & = & \big( \frac{d}{d\xi} r_t + \alpha(\omega_2,t,r_t) \big) dt + \sigma(r_t) dW_t \\ && + \int_G \gamma(r_{t-},x) \tilde{\mu}(dt,dx) + \int_I \delta(r_{t-},x) \mu^{\omega_2}(dt,dx) \medskip
\\ r_0 & = & h_0
\end{array}
\right.
\end{align}
on the filtered probability space $(\Omega_1,\mathcal{G},(\mathcal{G}_t)_{t \geq 0},\mathbb{Q}_1)$.
\end{theorem}

In order to prepare the proof of Theorem~\ref{4.3}, note that the drift term (\ref{SPDE:dc}) can be expressed as
\begin{align}
\alpha(\omega_2,t,h) = \alpha^1(h) + \alpha^2(h) + \alpha_{(\omega_2,t)}^3(h).
\end{align}
where we have set
\begin{align*}
\alpha^1(h) &:= \sum_{j \in \mathbb{N}} \sigma^j(h) \mathcal{I} \sigma^j(h),
\\ \alpha^2(h) &:= -\int_E \gamma(h,x) \exp(-\mathcal{I} \gamma(h,x)) \tilde{F}(dx),
\\ \alpha_{(\omega_2,t)}^3(h)(\xi,\eta) &:= -\int_0^{f_{(\omega_2,t)}(\eta)} \delta(h,x)(\xi,\eta) \exp(-\mathcal{I} \delta(h,x)(\xi,\eta)) \partial_x h(0,\omega_2(t)+x) dx,
\end{align*}
and where for each $(\omega_2,t) \in \Omega_2 \times \mathbb{R}_+$ the piecewise linear function $f_{(\omega_2,t)} : \mathbb{R}_+ \rightarrow [0,1]$ is defined as
\begin{align*}
f_{(\omega_2,t)}(\eta) :=
\begin{cases}
0, & \text{if $\eta - \omega_2(t-) \leq 0$,}
\\ \eta-\omega_2(t-), & \text{if $0 \leq \eta - \omega_2(t-) \leq 1 - \omega_2(t)$,}
\\ 1-\omega_2(t), & \text{if $\eta - \omega_2(t-) \geq 1 - \omega_2(t)$.}
\end{cases}
\end{align*}
Now, our goal is to show that $\alpha$ is locally Lipschitz and satisfies the linear growth condition. For this purpose, we prepare a few auxiliary results.

\begin{lem}\label{lemma-alpha-1}
The following statements are true:
\begin{enumerate}
\item We have $\alpha^1(H_{\beta}) \subset H_{\beta}$.

\item The mapping $\alpha^1 : H_{\beta} \rightarrow H_{\beta}$ is Lipschitz continuous.
\end{enumerate}
\end{lem}

\begin{proof}
According to Theorem~\ref{thm-group}, the multiplication $m : H_{\beta} \times H_{\beta} \rightarrow H_{\beta}$ is a continuous bilinear operator, the integral operator $\mathcal{I} : H_{\beta'}^0 \rightarrow H_{\beta}$ is a continuous linear operator, and we have $H_{\beta'} \subset H_{\beta}$ with $\| h \|_{\beta} \leq \| h \|_{\beta'}$ for all $h \in H_{\beta'}$. Thus, by (\ref{sigma-cond-1}), for all $h \in H_{\beta}$ we have
\begin{align*}
\| \alpha^1(h) \|_{\beta} &= \bigg\| \sum_{j \in \mathbb{N}} \sigma^j(h) \mathcal{I} \sigma^j(h) \bigg\|_{\beta} \leq \sum_{j \in \mathbb{N}} \| \sigma^j(h) \mathcal{I} \sigma^j(h) \|_{\beta}
\\ &\leq \| m \| \, \| \mathcal{I} \| \sum_{j \in \mathbb{N}} \| \sigma^j(h) \|_{\beta'}^2 \leq  \| m \| \, \| \mathcal{I} \| \sum_{j \in \mathbb{N}} (c^j)^2 < \infty,
\end{align*}
showing that $\alpha^1(H_{\beta}) \subset H_{\beta}$. Moreover, by (\ref{sigma-cond-1}), (\ref{sigma-cond-2}), for all $h_1,h_2 \in H_{\beta}$ we obtain
\begin{align*}
&\| \alpha^1(h_1) - \alpha^1(h_2) \|_{\beta} \leq \sum_{j \in \mathbb{N}} \| \sigma^j(h_1) \mathcal{I} \sigma^j(h_1) - \sigma^j(h_2) \mathcal{I} \sigma^j(h_2) \|_{\beta}
\\ &\leq \sum_{j \in \mathbb{N}} \| \sigma^j(h_1) (\mathcal{I} \sigma^j(h_1) - \mathcal{I} \sigma^j(h_2)) \|_{\beta} + \sum_{j \in \mathbb{N}} \| (\sigma^j(h_1) - \sigma^j(h_2)) \mathcal{I} \sigma^j(h_2) \|_{\beta}
\\ &\leq \| m \| \, \| \mathcal{I} \| \sum_{j \in \mathbb{N}} \| \sigma^j(h_1) \|_{\beta'} \| \sigma^j(h_1) - \sigma^j(h_2) \|_{\beta'} 
\\ &\quad + \| m \| \, \| \mathcal{I} \| \sum_{j \in \mathbb{N}} \| \sigma^j(h_1) - \sigma^j(h_2) \|_{\beta'} \| \sigma^j(h_2) \|_{\beta'}
\\ &\leq 2 \| m \| \, \| \mathcal{I} \| \bigg( \sum_{j \in \mathbb{N}} (c^j)^2 \bigg) \| h_1 - h_2 \|_{\beta},
\end{align*}
showing that $\alpha^1$ is Lipschitz continuous.
\end{proof}

\begin{lem}\label{lemma-alpha-2}
The following statements are true:
\begin{enumerate}
\item We have $\alpha^2(H_{\beta}) \subset H_{\beta}$.

\item The mapping $\alpha^2 : H_{\beta} \rightarrow H_{\beta}$ is Lipschitz continuous.
\end{enumerate}
\end{lem}

\begin{proof}
According to Theorem~\ref{thm-group}, the multiplication $m : H_{\beta} \times H_{\beta} \rightarrow H_{\beta}$ is a continuous bilinear operator, the integral operator $\mathcal{I} : H_{\beta'}^0 \rightarrow H_{\beta}$ is a continuous linear operator, and we have $H_{\beta'} \subset H_{\beta}$ with $\| h \|_{\beta} \leq \| h \|_{\beta'}$ for all $h \in H_{\beta'}$. Thus, by estimates (\ref{exp-lin-growth}) and (\ref{gamma-cond-1}), for all $h \in H_{\beta}$ we have
\begin{align*}
\| \alpha^2(h) \|_{\beta} &= \bigg\| \int_G \gamma(h,x) \exp(-\mathcal{I} \gamma(h,x)) \tilde F(dx) \bigg\|_{\beta} 
\\ &\leq \int_G \| \gamma(h,x) \exp(-\mathcal{I} \gamma(h,x)) \|_{\beta} \tilde F(dx)
\\ &\leq \| m \| \| C_2 \int_G \| \gamma(h,x) \|_{\beta} (1 + \| \mathcal{I }\gamma(h,x) \|_{\beta}) \exp(C_3 \| \mathcal{I} \gamma(h,x) \|_{\beta}) \tilde F(dx)
\\ &\leq \| m \| C_2 \int_G \| \gamma(h,x) \|_{\beta'} (1 + \| \mathcal{I} \| \, \| \gamma(h,x) \|_{\beta'}) \exp(C_3 \| \mathcal{I}  \| \, \| \gamma(h,x) \|_{\beta'}) \tilde F(dx)
\\  &\leq \| m \| C_2 M(1 + \| \mathcal{I} \| M) \exp(C_3 \| \mathcal{I} \| M) \tilde F(G) < \infty,
\end{align*}
showing that $\alpha^2(H_{\beta}) \subset H_{\beta}$. Moreover, by Theorem~\ref{thm-group} the mapping $h \mapsto \exp(h)$ is locally Lipschitz, and hence, in view of (\ref{gamma-cond-1}), there exists a constant $L \geq 0$ such that
\begin{align*}
\| \exp(h_1) - \exp(h_2) \|_{\beta} \leq L \| h_1 - h_2 \|_{\beta} \quad \text{for all $h_1,h_2 \in - \mathcal{I} \gamma(H_{\beta} \times G)$.}
\end{align*}
Therefore, by estimates (\ref{exp-lin-growth}) and (\ref{gamma-cond-1}), (\ref{gamma-cond-2}), for all $h_1,h_2 \in H_{\beta}$ we obtain
\begin{align*}
&\| \alpha^2(h_1) - \alpha^2(h_2) \|_{\beta}
\\ &\leq \int_G \| \gamma(h_1,x) ( \exp(-\mathcal{I} \gamma(h_1,x)) - \exp(-\mathcal{I} \gamma(h_2,x)) \|_{\beta} \tilde F(dx)
\\ &\quad + \int_G \| (\gamma(h_1,x) - \gamma(h_2,x))\exp(-\mathcal{I} \gamma(h_2,x)) \|_{\beta} \tilde F(dx)
\\ &\leq \int_G \| \gamma(h_1,x) \|_{\beta'} L \| \mathcal{I} \gamma(h_1,x) - \mathcal{I}\gamma(h_2,x) \|_{\beta'} \tilde F(dx) 
\\ &\quad + \int_G M \| h_1 - h_2 \|_{\beta} C_2 (1 + \| \mathcal{I} \gamma(h_2,x) \|_{\beta}) \exp(C_3 \| \mathcal{I} \gamma(h_2,x) \|_{\beta}) \tilde F(dx)
\\ &\leq \big( M^2 L \| \mathcal{I} \| + M C_2(1 + \| \mathcal{I} \| M) \exp(C_3 \| \mathcal{I} \| M) \big) \tilde F(G) \| h_1 - h_2 \|_{\beta},
\end{align*}
showing that $\alpha^2$ is Lipschitz continuous.
\end{proof}

In order to treat the mapping $\alpha^3$, we prepare a separate auxiliary result. For $(\omega_2,t) \in \Omega_2 \times \mathbb{R}_+$ and any bounded, measurable mapping $\epsilon : I \rightarrow H_{\beta}$ we define
\begin{align*}
\alpha_{(\omega_2,t)}^{\epsilon}(h)(\xi,\eta) := \int_0^{f_{(\omega_2,t)}(\eta)} \epsilon(x)(\xi,\eta) \partial_x h(0,\omega_2(t)+x) dx
\end{align*}
for $h \in H_{\beta}$ and $(\xi,\eta) \in \mathbb{R}_+ \times [0,1]$.

\begin{lem}\label{lemma-local-essential}
The following statements are true:
\begin{enumerate}
\item For all $(\omega_2,t) \in \Omega_2 \times \mathbb{R}_+$ and any bounded, measurable mapping $\epsilon : I \rightarrow H_{\beta}$ we have $\alpha_{(\omega_2,t)}^{\epsilon}(H_{\beta}) \subset H_{\beta}$. 

\item There exists a constant $N > 0$ such that for all $(\omega_2,t) \in \Omega_2 \times \mathbb{R}_+$ and any bounded, measurable mapping $\epsilon : I \rightarrow H_{\beta}$ we have
\begin{align}\label{est-pre-loc-Lipschitz}
\| \alpha_{(\omega_2,t)}^{\epsilon}(h) \|_{\beta} \leq N \| \epsilon \|_{\infty} \| h \|_{\beta}, \quad h \in H_{\beta}.
\end{align}
\end{enumerate}
\end{lem}

\begin{proof}
We fix $(\omega_2,t) \in \Omega_2 \times \mathbb{R}_+$ and a bounded, measurable mapping $\epsilon : I \rightarrow H_{\beta}$. Furthermore, let $h \in H_{\beta}$ be arbitrary. Let us determine the partial derivatives of $\alpha_{(\omega_2,t)}^{\epsilon}(h)$. The partial derivative $\partial_{\xi}$ is given by
\begin{align*}
\partial_{\xi} \alpha_{(\omega_2,t)}^{\epsilon}(h)(\xi,\eta) &= \int_0^{f_{(\omega_2,t)}(\eta)} \partial_{\xi} \epsilon(x)(\xi,\eta) \partial_x h(0,\omega_2(t) + x) dx,
\end{align*}
the partial derivative $\partial_{\eta}$ is given by
\begin{align*}
\partial_\eta \alpha_{(\omega_2,t)}^{\epsilon}(h)(\xi,\eta) &= \epsilon(f_{(\omega_2,t)}(\eta))(\xi,\eta) \partial_x h(0,\omega_2(t) + f_{(\omega_2,t)}(\eta))
\\ &\quad + \int_0^{f_{(\omega_2,t)}(\eta)} \partial_{\eta} \epsilon(x)(\xi,\eta) \partial_x h(0,\omega_2(t) + x) dx,
\end{align*}
and the second order derivative $\partial_{\xi \eta}$ is given by
\begin{align*}
\partial_{\xi \eta} \alpha_{(\omega_2,t)}^{\epsilon}(h)(\xi,\eta) &= \partial_{\xi} \epsilon(f_{(\omega_2,t)}(\eta))(\xi,\eta) \partial_x h(0,\omega_2(t) + f_{(\omega_2,t)}(\eta))
\\ &\quad +\int_0^{f_{(\omega_2,t)}(\eta)} \partial_{\xi \eta} \epsilon(x)(\xi,\eta) \partial_x h(0,\omega_2(t) + x) dx.
\end{align*}
In particular, we have
\begin{align*}
\partial_\xi \alpha_{(\omega_2,t)}^{\epsilon}(h)(\xi,0) = 0.
\end{align*}
By estimate (\ref{est-C4}) of Theorem~\ref{thm-group} we obtain
\begin{align*}
&\int_0^1 | \epsilon(f_{(\omega_2,t)}(\eta))(0,\eta) \partial_x h(0,\omega_2(t) + f_{(\omega_2,t)}(\eta)) |^2 d\eta 
\\ &\leq \sup_{x \in I} \| \epsilon(x)(0,\cdot) \|_{L^{\infty}([0,1])}^2  \int_0^1 | \partial_\eta h(0,\omega_2(t) + f_{(\omega_2,t)}(\eta)) |^2 d\eta
\\ &\leq C_1^2 \sup_{x \in I} \| \epsilon(x) \|_{\beta}^2 \int_0^1 | \partial_{\eta} h(0,\eta) |^2 d\eta \leq C_1^2 \| \epsilon \|_{\infty}^2 \| h \|_{\beta}^2.
\end{align*}
Moreover, we get
\begin{align*}
&\int_0^1 \bigg| \int_0^{f_{(\omega_2,t)}(\eta)} \partial_{\eta} \epsilon(x)(0,\eta) \partial_x h(0,\omega_2(t)+x) dx \bigg|^2 d\eta
\\ &\leq \int_0^1 \int_0^{1-\omega_2(t)} | \partial_{\eta} \epsilon(x)(0,\eta) \partial_x h(0,\omega_2(t)+x) |^2 dx d\eta
\\ &= \int_0^{1-\omega_2(t)} | \partial_x h(0,\omega_2(t) + x) |^2 \int_0^1 | \partial_{\eta} \epsilon(x)(0,\eta) |^2 d\eta dx 
\\ &\leq \int_0^{1-\omega_2(t)} | \partial_x h(0,\omega_2(t) + x) |^2 \| \epsilon(x) \|_{\beta}^2 dx
\leq \| \epsilon \|_{\infty}^2 \int_0^1 | \partial_{\eta} h(0,\eta) |^2 d\eta \leq \| \epsilon \|_{\infty}^2 \| h \|_{\beta}^2.
\end{align*}
For every fixed $\eta \in [0,1]$ we have
\begin{align*}
\partial_{\xi} \epsilon(f_{(\omega_2,t)}(\eta))(\xi,\eta) = \partial_{\xi} \epsilon(f_{(\omega_2,t)}(\eta))(\xi,0) + \int_0^{\eta} \partial_{\xi \eta} \epsilon(f_{(\omega_2,t)}(\eta))(\xi,\bar{\eta}) d \bar{\eta},
\end{align*}
which implies
\begin{align*}
&\int_0^{\infty} | \partial_{\xi} \epsilon(f_{(\omega_2,t)}(\eta))(\xi,\eta) |^2 e^{\beta \xi} d\xi 
\\ &\leq 2 \int_0^{\infty} \bigg| \partial_{\xi} \epsilon(f_{(\omega_2,t)}(\eta))(\xi,0) \bigg|^2 e^{\beta \xi} d\xi + 2 \int_0^{\infty} \bigg| \int_0^{\eta} \partial_{\xi \eta} \epsilon(f_{(\omega_2,t)}(\eta))(\xi,\bar{\eta}) d \bar{\eta} \bigg|^2 e^{\beta \xi} d\xi
\\ &\leq 2 \| \epsilon(f_{(\omega_2,t)}(\eta)) \|_{\beta}^2 + 2 \int_0^{\infty} \int_0^{1} | \partial_{\xi \eta} \epsilon(f_{(\omega_2,t)}(\eta))(\xi,\bar{\eta}) |^2 d \bar{\eta} e^{\beta \xi} d\xi
\\ &\leq 4 \| \epsilon(f_{(\omega_2,t)}(\eta)) \|_{\beta}^2 \leq 4 \| \epsilon \|_{\infty}^2.
\end{align*}
Therefore, we obtain
\begin{align*}
&\int_0^{\infty} \int_0^1 | \partial_{\xi} \epsilon(f_{(\omega_2,t)}(\eta))(\xi,\eta) \partial_x h(0,\omega_2(t) + f_{(\omega_2,t)}(\eta)) |^2 d\eta e^{\beta \xi} d\xi
\\ &\leq \int_0^1 | \partial_x h(0,\omega_2(t) + f_{(\omega_2,t)}(\eta)) |^2 \int_0^{\infty} | \partial_{\xi} \epsilon(f_{(\omega_2,t)}(\eta))(\xi,\eta) |^2 e^{\beta \xi} d\xi d\eta
\\ &\leq 4 \| \epsilon \|_{\infty}^2 \int_0^1 | \partial_{\eta} h(0,\eta) |^2 d\eta \leq 4 \| \epsilon \|_{\infty}^2 \| h \|_{\beta}^2.
\end{align*}
Moreover, we have
\begin{align*}
&\int_0^{\infty} \int_0^1 \bigg| \int_0^{f_{(\omega_2,t)}(\eta)} \partial_{\xi \eta} \epsilon(x)(\xi,\eta) \partial_x h(0,\omega_2(t) + x) dx \bigg|^2 d\eta e^{\beta \xi} d\xi
\\ &\leq \int_0^{\infty} \int_0^1 \int_0^{1 - \omega_2(t)} | \partial_{\xi \eta} \epsilon(x)(\xi,\eta) \partial_x h(0,\omega_2(t) + x) |^2 dx d\eta e^{\beta \xi} d\xi
\\ &\leq \int_0^{1-\omega_2(t)} | \partial_x h(0,\omega_2(t)+x) |^2 \int_0^{\infty} \int_0^1 | \partial_{\xi \eta} \epsilon(x)(\xi,\eta) |^2 d\eta e^{\beta \xi} d\xi dx
\\ &\leq \int_0^{1-\omega_2(t)} | \partial_x h(0,\omega_2(t)+x) |^2 \| \epsilon(x) \|_{\beta}^2 dx
\leq \| \epsilon \|_{\infty}^2 \int_0^1 | \partial_{\eta} h(0,\eta) |^2 d\eta \leq \| \epsilon \|_{\infty}^2 \| h \|_{\beta}^2.
\end{align*}
Taking into account the definition (\ref{def-norm-beta}) of the norm $\| \cdot \|_{\beta}$, this concludes the proof.
\end{proof}

Now, we shall prove that $\alpha$ is locally Lipschitz and satisfies the linear growth condition.

\begin{prop}\label{prop-alpha-3}
The following statements are true:
\begin{enumerate}
\item We have $\alpha(\Omega_2 \times \mathbb{R}_+ \times H_{\beta}) \subset H_{\beta}$.

\item For each $n \in \mathbb{N}$ there exists a constant $L_n \geq 0$ such that
\begin{align*}
\| \alpha(\omega_2,t,h_1) - \alpha(\omega_2,t,h_2) \|_{\beta} \leq L_n \| h_1 - h_2 \|
\end{align*}
for all $(\omega_2,t) \in \Omega_2 \times \mathbb{R}_+$ and all $h_1,h_2 \in H_{\beta}$ with $\| h_1 \|_{\beta}, \| h_2 \|_{\beta} \leq n$.

\item There exists a constant $K \geq 0$ such that
\begin{align*}
\| \alpha(\omega_2,t,h) \|_{\beta} \leq K \| h \|_{\beta} \quad \text{for all $(\omega_2,t,h) \in \Omega_2 \times \mathbb{R}_+ \times H_{\beta}$.}
\end{align*}
\end{enumerate}
\end{prop}

\begin{proof}
We define the mapping $\epsilon$ as
\begin{align*}
\epsilon(h,x) := \delta(h,x) \exp(-\mathcal{I} \delta(h,x)), \quad (h,x) \in H_{\beta} \times I.
\end{align*}
As in the proof of Lemma~\ref{lemma-alpha-2}, we show that $\epsilon(H_{\beta} \times I) \subset H_{\beta}$, and that there is a constant $M_{\epsilon} > 0$ such that for all $x \in I$ we have
\begin{align*}
\| \epsilon(h,x) \|_{\beta} &\leq M_{\epsilon}, \quad h \in H_{\beta}
\\ \| \epsilon(h_1,x) - \epsilon(h_2,x) \|_{\beta} &\leq M_{\epsilon} \| h_1 - h_2 \|_{\beta}, \quad h_1,h_2 \in H_{\beta}.
\end{align*}
In particular, for each $h \in H_{\beta}$ the mapping $\epsilon(h,\cdot) : I \rightarrow H_{\beta}$ is bounded and measurable. Moreover, we have
\begin{align*}
\alpha_{(\omega_2,t)}^3(h) = \alpha_{(\omega_2,t)}^{\epsilon(h,\cdot)} \quad \text{for all $(\omega_2,t,h) \in \Omega_2 \times \mathbb{R}_+ \times H_{\beta}$.}
\end{align*}
Let $(\omega_2,t) \in \Omega_2 \times \mathbb{R}_+$ be arbitrary. By Lemma~\ref{lemma-local-essential} we have $\alpha_{(\omega_2,t)}^3(H_{\beta}) \subset H_{\beta}$, for each $h \in H_{\beta}$  we have
\begin{align*}
\| \alpha_{(\omega_2,t)}^3(h) \|_{\beta} = \| \alpha_{(\omega_2,t)}^{\epsilon(h,\cdot)}(h) \|_{\beta} \leq N \| \epsilon(h,\cdot) \|_{\infty} \| h \|{\beta} \leq N M_{\epsilon} \| h \|_{\beta},
\end{align*}
and for all $h_1,h_2 \in H_{\beta}$ we obtain
\begin{align*}
&\| \alpha_{(\omega_2,t)}^3(h_1) - \alpha_{(\omega_2,t)}^3(h_2) \|_{\beta} = \| \alpha_{(\omega_2,t)}^{\epsilon(h_1,\cdot)}(h_1) - \alpha_{(\omega_2,t)}^{\epsilon(h_2,\cdot)}(h_2) \|_{\beta}
\\ &\leq \| \alpha_{(\omega_2,t)}^{\epsilon(h_1,\cdot) - \epsilon(h_2,\cdot)}(h_1) \|_{\beta} + \| \alpha_{(\omega_2,t)}^{\epsilon(h_2,\cdot)}(h_1-h_2) \|_{\beta}
\\ &\leq N \| \epsilon(h_1,\cdot) - \epsilon(h_2,\cdot) \|_{\infty} \| h_1 \|_{\beta} + N \| \epsilon(h_2,\cdot) \|_{\infty} \| h_1 - h_2 \|_{\beta}
\\ &\leq N M_{\epsilon} \| h_1 - h_2 \|_{\beta} \| h_1 \|_{\beta} + N M_{\epsilon} \| h_1 - h_2 \|_{\beta}.
\end{align*}
Together with Lemmas~\ref{lemma-alpha-1} and \ref{lemma-alpha-2}, this concludes the proof.
\end{proof}

Now, we are ready to provide the proof of Theorem~\ref{4.3}.

\begin{proof}[Proof of Theorem~\ref{4.3}]
According to Theorem~\ref{thm-group} there exist another separable Hilbert space $\mathcal{H}_{\beta}$, a $C_0$-group $(U_t)_{t \in \mathbb{R}}$ on $\mathcal{H}_{\beta}$ and continuous linear operators $\ell \in L(H_{\beta},\mathcal{H}_{\beta})$, $\pi \in L(\mathcal{H}_{\beta},H_{\beta})$ such that $\pi U_t \ell = S_t$ for all $t \in \mathbb{R}_+$. Therefore, by virtue of condition (\ref{sigma-cond-2}) and Proposition~\ref{prop-alpha-3}, existence and uniqueness of mild solutions for the SPDE (\ref{SPDE-omega-2}) follows from \citeN[Theorem~4.5]{Tappe2012a}.
\end{proof}

\section{Positivity and monotonicity}\label{sec-pos-mon}

Monotonicity of the bond prices $P(t,T,\eta)$ with respect to the quality $\eta$ is a desirable modelling feature. As we shall see, it is implied by the positivity of the forward rates.

We continue to work under Assumption (A1') such that the loss process $L$ is given by \eqref{def:L2}. 
Moreover, we study the forward rates given by the SPDE in \eqref{Musiela-CDO} and assume that the drift condition is satisfied, i.e.~\eqref{SPDE:dc} holds.

\begin{defin}
The term structure model (\ref{eq:TxbondsViaForwardRates}) is called \emph{monotone}, if for all $0 \leq t \leq T$ and $ 0 \le \eta_1 \leq \eta_2 \le 1$ we have
$$ \mathbb{Q}\big(  P(t,T, \eta_1) \le P(t,T,\eta_2) \big) =1. $$
\end{defin}

Since we study the forward rate dynamics under a martingale measure $\Q$, the discounted bond prices are local martingales. If they are even true martingales, then monotonicity follows directly as a consequence of our upcoming result.

\begin{prop}\label{prop-MT-monoton}
Consider  $ 0 \le \eta_1 \leq \eta_2 \le 1$ and $T \geq 0$, and assume that 
$ (D_{t} P(t,T,\eta_i))_{0 \le t \le T}$ are martingales for $i=1,2$. Then for all $t \in [0,T]$ we have
$$ \mathbb{Q}\big(  P(t,T, \eta_1) \le P(t,T,\eta_2) \big) =1. $$
\end{prop}
\begin{proof}
Let $t \in [0,T]$ be arbitrary. By definition the discounting process $D$ is positive. The martingale property and representation (\ref{repr-bond}) yield that
$$ P(t,T,\eta_i) = \E^{\Q} \left[ \frac{D_T}{D_t} \ind{L_T \le \eta_i} \, \Big| \, \cF_t \right], \quad i=1,2. $$
Moreover, we have $\ind{y \le \eta_1} \le \ind{y \le \eta_2}$ for all $y\in\R$. Consequently, the monotonicity of the conditional expectation gives the result.  
\end{proof}

In order to define positivity of the forward rates, we introduce the closed, convex cone of non-negative functions of $H_{\beta}$ as
\begin{align*}
\cP = \{ h \in H_{\beta} : h(\xi,\eta) \ge 0 \text{ for all } (\xi,\eta) \in \R_+ \times [0,1] \}.
\end{align*}

\begin{defin}
The family of term structure models (\ref{Musiela-CDO}) is called \emph{positivity preserving} if for all $h_0 \in \cP$ we have
\begin{align*}
\mathbb{Q}(r_t \in \cP) = 1 \quad \text{for all $t \geq 0$,}
\end{align*}
where $(r_t)_{t \geq 0}$ denotes the mild solution for the SPDE (\ref{Musiela-CDO}) with $r_0 = h_0$.
\end{defin}

Now, we will prove that the positivity preserving property implies the monotonicity of the term structure model.

\begin{prop}\label{prop-pos-mon}
If the family of term structure models (\ref{Musiela-CDO}) is \emph{positivity preserving}, then for each $h_0 \in \cP$ the term structure model (\ref{eq:TxbondsViaForwardRates}) is monotone.
\end{prop}

\begin{proof}
Let $h_0 \in \cP$ be arbitrary and denote by $(r_t)_{t \geq 0}$ denotes the mild solution for the SPDE (\ref{Musiela-CDO}) with $r_0 = h_0$. By the positivity preserving property we have 
$$
\Q(r_t(\xi,\eta) \geq 0 \text{ for all } (\xi,\eta) \in \R_+ \times [0,1] ) = 1 \quad \text{for all $t \geq 0$.}
$$
Let $(T,\eta) \in \R_+ \times [0,1]$ be arbitrary.
We will show that the discounted $(T,\eta)$-bond price process is a true martingale: First, \eqref{dc1:2} and \eqref{dc2:2} are satisfied and Proposition~\ref{prop:2.1} gives that the process
$(D_t P(t,T,\eta))_{0 \leq t \leq T}$ is a local martingale. Moreover, by the representation (\ref{repr-bond}) we have
$$
D_t P(t,T,\eta) = e^{-\int_0^t r_u(0,1) du} \ind{L_t \le \eta} e^{-\int_0^{T-t} r_t(u,\eta)du}, \quad t \in [0,T].
$$
Therefore, we obtain
$$ 0 \le D_t P(t,T,\eta) \le \ind{L_t \le \eta} \le 1, \quad t \in [0,T],$$
and hence $(D_t P(t,\xi,\eta))_{0 \le t \le T}$ is a true martingale. Applying Proposition~\ref{prop-MT-monoton} finishes the proof.
\end{proof}

Now, we shall derive conditions for the positivity preserving property in terms of the characteristic coefficients $\sigma$, $\gamma$ and $\delta$ of the SPDE (\ref{Musiela-CDO}).

\begin{thm}\label{thm:positivity}
Suppose $\sigma \in C^2(H_{\beta};L_2^0(H_{\beta}))$ and that the mapping $$h \mapsto \sum_{j \in \mathbb{N}} D \sigma^j(h) \sigma^j(h)$$ 
is Lipschitz continuous on $H_{\beta}$. Furthermore, we assume that for all $j \in \mathbb{N}$ we have
\begin{align}\label{cond-sigma}
\sigma^j(h)(\xi,\eta) = 0, \quad \text{for all $h \in \cP$ and $(\xi,\eta) \in \mathbb{R}_+ \times [0,1]$ with $h(\xi,\eta) = 0$}
\end{align}
and for $\tilde F$-almost all $x \in G$ and all $y \in I$ we have
\begin{align}\label{inv-2}
&h + \gamma(h,x) + \delta(h,y) \in \cP, \quad \text{for all $h \in \cP$,}
\\ \label{inv-3} &\gamma(h,x)(\xi,\eta) = 0, \quad \text{for all $h \in \cP$ and $(\xi,\eta) \in \mathbb{R}_+ \times [0,1]$ with $h(\xi,\eta) = 0$,}
\\ \label{inv-4} &\delta(h,y)(\xi,\eta) = 0, \quad \text{for all $h \in \cP$ and $(\xi,\eta) \in \mathbb{R}_+ \times [0,1]$ with $h(\xi,\eta) = 0$.}
\end{align}
Then, the family of term structure models (\ref{Musiela-CDO}) is positivity preserving.
\end{thm}

\begin{proof}
Since the drift term $\alpha : \Omega_2 \times \R_+ \times H_{\beta} \rightarrow H_{\beta}$ of the SPDE (\ref{Musiela-CDO}) is given by (\ref{SPDE:dc}), conditions (\ref{cond-sigma}), (\ref{inv-3}) and (\ref{inv-4}) yields that for all $(\omega_2,t) \in \Omega_2 \times \R_+$ we have
\begin{align*}
\alpha(\omega_2,t,h)(\xi,\eta) = 0, \quad \text{for all $h \in \cP$ and $(\xi,\eta) \in \mathbb{R}_+ \times [0,1]$ with $h(\xi,\eta) = 0$.}
\end{align*}
Thus, proceeding as in \citeN[Section~4]{FilipovicTappeTeichmann2010} gives the positivity preserving property of the family of term structure models (\ref{Musiela-CDO}).
\end{proof}

\section{Examples}\label{examples}
In this section we illustrate applications of our general approach. We first discuss the application to the modelling of collateralized debt obligations, thus generalizing \citeN{FilipovicSchmidtOverbeck} and \citeN{SchmidtZabczyk12}. 
Thereafter we consider less specific models for portfolio credit risk in a top-down setting and related infinite-dimensional credit risk models. Finally, we give an account of stochastic mortality modelling following \citeN{TappeWeber2013}.
\subsection{CDO term structure modelling}\label{sec:examples}
In this section the general modelling of credit risk markets by $(T,\eta)$-bonds is applied to a fixed and finite portfolio
of credit risky instruments. Typical derivatives in this markets are collateralized debt obligations (CDOs) and 
single-tranche CDOs. We show how the top-down approaches for CDO markets can be embedded in the more general setting considered here.  

Mathematically speaking, a CDO is  a derivative on a portfolio of $N$ credit risky instruments.
With each instrument there is an associated nominal and we assume that the total nominal is one.
Denote the process of accumulated losses over time by $L=\stpr{L}$.
This is a pure-jump process which jumps upward at default of instruments in the pool by the occurring loss.
As the total nominal is one,  $L_t \in [0,1]$ for all $t \ge 0$.
By $\cI:=[0,1]$ we denote the set of attainable loss fractions, the case where $\cI$ is finite may be considered analogously.

To facilitate the mathematical analysis of CDO markets, we introduce  $(T,\eta)$-bonds.
A $(T,\eta)$-bond pays $\ind{L_t \le \eta}$ at maturity $T$. Hence, in our setting $\{\tau_\eta > t\} = \{L_t \le \eta\}$.
For $\eta=1$ we obtain that $P(t,T,1)=:P(t,T)$ which equals the risk-free bond.

A securitization mechanism splits the CDO pool in several tranches which have different risk profiles and serve as efficient instrument to enable trading on the CDO pool. The single-tranche CDO (STCDO) is specified by
\begin{itemize}\setlength{\itemindent}{-5mm} 
\item a number of future trading dates $T_0 < T_1 < \cdots < T_n$,
\item lower and upper detachment points $x_1, x_2$ specifying the \emph{tranche} $(x_1,x_2]\subset \cI$,
\item a fixed swap rate $\kappa$.
\end{itemize}
Set 
$$ H(x):= (x_2-x)^+ - (x_1-x)^+ = \int_{(x_1,x_2]} \ind{x \le y} dy. $$
Then the payment scheme of the STCDO can be described as follows: an investor in this STCDO
\begin{itemize}\setlength{\itemindent}{-5mm} 
\item receives $\kappa H(L_{T_i})$ at $T_i$, $i=1,\dots,n$,
\item pays $H(L_{t-})-H(L_t)$ at any time $t \in (T_0,T_n]$ when $\Delta L_t \not =0$ (i.e.~when a default occurs).
\end{itemize} 
It has been shown in \citeN[Lemma 4.1]{FilipovicSchmidtOverbeck} that by a Fubini-type argument prices of STCDOs can be expressed directly in terms of prices of $(T,\eta)$-bonds.

Regarding absence of arbitrage, we assume that $L$ satisfies
\begin{align} \label{def:L}
 L_t = \int_0^t \int_E \ind{L_{s-}+\ell_s(x)\le 1}\ell_s(x) \mu(ds,dx),
 \end{align}
where $\ell$ is a non-negative, predictable process such that for all $t \ge 0$ it holds that $\int_0^t \ind{L_{s-}+\ell_s(x)\le 1}\ell_s(x) F_s(dx) ds < \infty$ (finite activity).

Then  $L$ is a non-decreasing, pure-jump process with values in $\cI$. Furthermore,
the indicator process $(\ind{L_t\le \eta})_{t \ge 0}$ is c\`{a}dl\`{a}g and has intensity
\begin{equation}\label{deflambda}
  \lambda(t,\eta):= F_t ( \{x \in E: L_{t-} + \ell_t(x) > \eta \} );
\end{equation}
that is,
\begin{equation} \label{M2}
    M_t:= 1_{\{ L_t\le \eta\}} +\int_0^t 1_{\{ L_{s}\le \eta\}}     \lambda(s,\eta)\,ds
\end{equation}
is a martingale. Moreover, $\lambda(t,\eta)$ is decreasing in   $\eta$ with $ \lambda(t,1)=0$.
With $\tau_\eta := \inf\{t \ge 0: L_t >\eta \}$ we obtain the final link to $(T,\eta)$-bonds as in \eqref{repr-bond}:
\begin{align*}
P(t,T,\eta) = \ind{L_t \le \eta} \exp \bigg( - \int_t^{T} f(t,u,\eta) du\bigg).
\end{align*}

As a corollary we obtain the generalization of the drift condition to the infinite-dimensional setup considered in our article.
The result directly follows from Proposition \ref{prop:2.1}.
\begin{corollary}\label{cor-drift-cond}
Assume that (A2)--(A5) and \eqref{def:L} hold. Then $\Q$ is ELMM, if and only if
\begin{align}
      \alpha(t,T,\eta)  &= \sum_{j \in \mathbb{N}} \sigma^j (t,T,\eta) \Sigma^j(t,T,\eta) \nonumber\\
&- \int _{E} \gamma(t,T,\eta,x) \Big( e^{-\Gamma(t,T,\eta,x) }\ind{L_{t-}+\ell_t(x)\le \eta}-1 \Big)  F_t(dx)
 \label{dc1a}\\
    r_t(0,\eta) &= r_t+\lambda(t,\eta) , \label{dc2a}
\end{align}
where \eqref{dc1a} and \eqref{dc2a} hold on $\{L_t \le \eta\}$, $\Q\otimes dt$-a.s.
\end{corollary}

Next, we shall discuss conditions for positivity, and hence monotonicity, of the model. If for all $j \in \mathbb{N}$ we have
\begin{align*}
\sigma^j(h)(\xi,\eta) = 0, \quad \text{for all $h \in \cP$ and $(\xi,\eta) \in \mathbb{R}_+ \times [0,1]$ with $h(\xi,\eta) = 0$}
\end{align*}
and for all $x \in E$ we have
\begin{align*}
&h + \gamma(h,x) \in \cP, \quad \text{for all $h \in \cP$,}
\\ &\gamma(h,x)(\xi,\eta) = 0, \quad \text{for all $h \in \cP$ and $(\xi,\eta) \in \mathbb{R}_+ \times [0,1]$ with $h(\xi,\eta) = 0$,}
\end{align*}
then the positivity preserving property is fulfilled, which follows from Theorem \ref{thm:positivity}.

\subsection{Top-down modelling of credit portfolios}
As next example we  specify a class of models where we consider the ordered default times. This is often called a top-down approach and simplifies the analysis of the model. 
In this regard,  consider a portfolio of $N$ credit risky instruments. We denote their default times by $\sigma_1,\dots,\sigma_N$. Define the associated counting process by
$$ L_t := \frac{1}{N}\sum_{i=1}^N \ind{\sigma_i \le t}. $$
Letting $\cI=\{0,N^{-1},\dots,1\}$ we obtain that the associated first hitting times of $L$ with barrier $\eta \in \cI$,
$$ \tau_\eta := \inf\{ t \ge 0: L_t \ge \eta \} $$
equal the ordered default times $\sigma_{(1)},\dots,\sigma_{(N)}$ (if no joint defaults happen). Hence the whole portfolio and the associated loss process can be modelled by looking at term structures of $(T,\eta)$-bonds and the ordered default times $\{ \tau_\eta: \eta \in \cI \}$.

\subsection{Infinite dimensional markets with credit risk}
\label{sec:infinitebondmarket}

In this section we introduce a general model for
a large financial market bearing credit risk. We basically construct an infinite-dimensional
intensity based model. Intensity based models have been intensively studied in the literature, see \citeN{BieleckiRutkowski2002} or \citeN{Filipovic2009} for details and references. 

Consider a market with countably many companies and set $\cI=\NN$. Each company $\eta\in \cI$ is subject to default risk and we denote its default 
time by $\tau_\eta$. Each $\tau_\eta$ is assumed to be an $(\cF_t)$-stopping time. We assume that for each $\eta$ there exists a non-negative, optional process $(\lambda(t,\eta))_{t \ge 0}$ such that
$$ 
  M^\eta(t):=  \ind{\tau_\eta >t } + \int _0^t \ind{\tau_\eta \ge s} \lambda(s,\eta) ds
$$
is a martingale. The process $(\lambda(t,\eta))$ is called default intensity of company $\eta$. 

We associate a random measure $\mu^*$ to the default times as follows: the mark space $F=\{0,1\}^\infty$ is spanned by the unit vectors $e_1,e_2,\dots$ and
$$ \mu^*(A \times B) := \sum_{\eta \in \cI: e_\eta \in B} \delta_{\tau_\eta} (A),  \qquad A \in \cB(\R_+),\ B \subset F,$$
where $\delta_t$ denotes the Dirac-measure at time $t$.  We define 
$$ F_t^*(B) :=   \sum_{\eta \in \cI: e_\eta \in B}  \lambda(t,\eta) $$
and obtain that $\mu^*(dt,dx)-F_t^*(dx)dt$ is a compensated Poisson random measure.

Assume that $E=F \times G$, where  $G$ is the mark space of a homogeneous Poisson random measure $\tilde \mu$. 
We set $\mu=\mu^* \otimes \tilde \mu$ and obtain a model which is driven by the defaults and possible further jumps from $\tilde \mu$. Note that (A1) is satisfied: 
\begin{align*} 
  \ind{\tau_\eta > t} &= 1-  \ind{\tau_\eta \le t} = 1- \mu^*([0,t] \times \{e_\eta\}) 
    \\
    &= 1 - \int_0^t \ind{\tau_\eta \ge s} \int_E \ind{F \times \{e_\eta\}}(x) \mu(ds,dx),
\end{align*} 
and choosing $\beta(s,\eta,x) = - \ind{ F \times \{e_\eta\} }(x)$ yields the desired representation \eqref{eq:tau}.

The market trades bonds for each companies whose prices are denoted by $P(t,T,\eta)$ and we assume that they follow the HJM-representation in terms of forward rates given in \eqref{eq:TxbondsViaForwardRates}. This model is included in our setup such that all our results can be applied. In particular, the default intensity of a single company can depend of the number of defaults in the market in the past, such that features like self-excitement can be included (see \citeN{GieseckeSpiliopoulosSowers2013} for an account of finite-dimensional markets with this feature).

\subsection{Stochastic mortality modelling}

In this section, we shall briefly illustrate our developed methods concerning positivity and monotonicity from Section~\ref{sec-pos-mon} in the context of stochastic mortality models. In the sequel, we follow the framework of \citeN{TappeWeber2013}, to which we refer for further details.

The death of an individual born at $-\eta\in\R_{\le 0}=\cI$ is denoted by a $\mathcal{G}$-measurable random time $\tau^\eta : \Omega \rightarrow (-\eta,\infty)$ for some larger $\sigma$-algebra $\mathcal{G} \supset \mathcal{F}$. The filtration $(\mathcal{F}_t)_{t \geq 0}$, which is called the background information, contains all information about the likelihoods, but no information about the exact times of death events. We define the \emph{survival process}
\begin{align*}
G(t,t,\eta) := \mathbb{P}(\tau^\eta > t \,|\, \mathcal{F}_t),
\end{align*}
and for $t \leq T$ we define the \emph{forward survival process}
\begin{align*}
G(t,T,\eta) := \mathbb{P}(\tau^\eta > T \,|\, \mathcal{F}_t) = \mathbb{E}[G(T,T,\eta)  \,|\, \mathcal{F}_t]
\end{align*}
as the best prediction at date $t$ of the fraction of individuals born at date $-\eta$ that survive until a future date $T$. 
Then the forward survival processes become martingales which allows us to relate this approach to our setup where NAFL was equivalent to considering equivalent local martingale measures.

In this regard, we may proceed as follows: after performing the Musiela type change of parameters $(t,T,\eta) \mapsto (t,T-t,\eta+t) =: (t,x,\xi)$, we can model the mortality rates $\mu$ as a SPDE
\begin{align*}
d \mu_t = \big( (\partial_x - \partial_\xi) \mu_t + \alpha(\mu_t) \big) dt + \sigma(\mu_t)dW_t + \int_E \delta(\mu_{t-},y) (\mu(dt,dy) - \nu(dy)dt),
\end{align*}
where the state space $H$ is a separable Hilbert space consisting of surfaces $h : \Xi \rightarrow \mathbb{R}$ with domain $\Xi = \{ (x,\xi) \in \mathbb{R}_+ \times \mathbb{R} : -\xi \leq x \}$. Forward survival process now play the role of bond prices in the following form 
$$ G(t,T,\eta) = F(\eta) \exp\Big(- \int_{0}^T  \mu_{s\wedge t}(s,\eta) ds\Big), \quad 0 \le t \le T. $$
Due to the martingale property of the survival processes, we obtain a drift condition in the form 
\begin{align*}
\alpha_t(x,\xi) &= \sum_{k \in \mathbb{N}} \sigma_t^k(x,\xi) \int_{0}^{x} \sigma_t^k(u,\xi) du
\\ &\quad - \int_E \delta_t(x,\xi,y) \bigg[ \exp \bigg( -\int_{0}^{x} \delta_t(u,\xi,y) du \bigg) - 1 \bigg] \nu(dy),
\end{align*}
similar to the HJM drift condition for default free bond markets.
As in Section~\ref{sec-pos-mon}, we can formulate appropriate conditions on $\alpha$, $\sigma$ and $\delta$ for the positivity preserving property of this SPDE. Then the survival processes calculated from these mortality rates satisfy $0 \leq G(t,T,\eta) \leq 1$. Therefore, they are not only local, but even true martingales, and hence, the mortality model is monotone by Proposition \ref{prop-MT-monoton}.


\end{document}